\newif\iflipics
\iflipics

\documentclass[a4paper,UKenglish,cleveref, autoref, thm-restate,anonymous]{lipics-v2021}
\title{Parameterized dynamic data structure for Split Completion} %TODO Please add
\author{Konrad Majewski}{Institute of Informatics, University of Warsaw, Poland}{k.majewski@mimuw.edu.pl}{}{}
\authorrunning{K. Majewski, Mi. Pilipczuk, A. Zych-Pawlewicz} %TODO mandatory. First: Use abbreviated first/middle names. Second (only in severe cases): Use first author plus 'et al.'

\Copyright{Konrad Majewski, Micha\l{} Pilipczuk, Anna Zych-Pawlewicz} %TODO mandatory, please use full first names. LIPIcs license is "CC-BY";  http://creativecommons.org/licenses/by/3.0/

\ccsdesc{Theory of computation~Parameterized complexity and exact algorithms}
\ccsdesc{Theory of computation~Data structures design and analysis}
\keywords{parameterized complexity, dynamic data structures, split graphs} %TODO mandatory; please add comma-separated list of keywords

\EventEditors{John Q. Open and Joan R. Access}
\EventNoEds{2}
\EventLongTitle{42nd Congress of Visibly Introvert Theoreticians (CVIT 2016)}
\EventShortTitle{CVIT 2016}
\EventAcronym{CVIT}
\EventYear{2016}
\EventDate{December 24--27, 2016}
\EventLocation{Little Whinging, United Kingdom}
\EventLogo{}
\SeriesVolume{42}
\ArticleNo{23}
\else % arxiv

\documentclass[11pt]{article}
\usepackage{a4wide}
\usepackage[utf8]{inputenc}
\usepackage[T1]{fontenc}

\title{Parameterized dynamic data structure for Split Completion
\thanks{This work is a part of project BOBR that has received funding from the European Research Council (ERC) under the European Union’s Horizon 2020 research and innovation programme (grant agreement No. 948057).}
}
\author{Konrad Majewski\thanks{Institute of Informatics, University of Warsaw, Poland (\texttt{k.majewski@mimuw.edu.pl})} \and
Michał Pilipczuk\thanks{Institute of Informatics, University of Warsaw, Poland (\texttt{michal.pilipczuk@mimuw.edu.pl})} \and
Anna Zych-Pawlewicz\thanks{Institute of Informatics, University of Warsaw, Poland (\texttt{anka@mimuw.edu.pl})}}

\fi

\iflipics
\else % arxiv
\usepackage{inconsolata}
\usepackage{libertine}
\fi

\usepackage{amsthm}
\usepackage{amssymb}
\usepackage{amsmath}
\usepackage{bbm} % indicator function
\usepackage{mathtools}
\usepackage{mathrsfs}
\usepackage{stmaryrd} % Iverson brackets
\usepackage{thmtools,thm-restate}

\usepackage{tikz}
\usepackage{graphicx}
\graphicspath{ {./images/} }

\usepackage{algorithm}
\usepackage[noend]{algpseudocode}
\usepackage{caption}
\usepackage{subcaption}
\usepackage[noadjust]{cite}
\usepackage{xspace}
\usepackage{xcolor}
\usepackage{xparse}
\usepackage{xfrac}
\usepackage{todonotes}
\usepackage{textpos}

\iflipics
\else % arxiv
\usepackage[colorlinks = true,linkcolor = blue,urlcolor = blue, citecolor = blue]{hyperref}
\usepackage[capitalize, nameinlink]{cleveref}
\fi

\iflipics
\else % arxiv
\newtheorem{theorem}{Theorem}[section]

\newtheorem{lemma}[theorem]{Lemma}

\crefname{claim}{Claim}{Claims}
\newtheorem{claim}[theorem]{Claim}
\newcommand{\cqed}{\ensuremath{\lhd}}
\makeatletter
\newenvironment{claimproof}{\par
	\pushQED{\cqed}%
	\normalfont \topsep6\p@\@plus6\p@\relax
	\trivlist
	\item\relax
	{\itshape
		Proof of the claim\@addpunct{.}}\hspace\labelsep\ignorespaces
}{%
	\hfill\popQED\endtrivlist\@endpefalse
}
\makeatother
\fi

\newtheorem{fact}[theorem]{Fact}

\newcommand{\Inline}[1]{#1\xspace}

\newcommand{\ceil}[1]{\left\lceil #1 \right\rceil}

\newcommand{\symd}{\triangle}

\newcommand{\splittance}{\mathsf{splittance}}
\newcommand{\oper}[1]{\mathtt{#1}}
\newcommand{\answer}[1]{\Inline{\textsf{#1}}}

\newcommand{\Non}{\overline{N}}

\newcommand{\Oh}{\mathcal{O}}
\newcommand{\Ohtilde}{\widetilde{\mathcal{O}}}

\newcommand{\indic}{\mathbbm{1}}

\newcommand{\E}{\mathbb{E}\,}

\newcommand{\N}{\mathbb{N}}

\newcommand{\Cc}{\mathcal{C}}
\newcommand{\Fc}{\mathcal{F}}

\newcommand{\Lc}{\mathcal{L}}

\renewcommand{\Pr}{\mathbb{P}}

\renewcommand{\leq}{\leqslant}
\renewcommand{\geq}{\geqslant}

\newcommand{\app}{\spadesuit}

\newcommand{\desc}[2]{{#1}^{\textrm{#2}}}
\newcommand{\Vmoved}{\desc{V}{moved}}
\newcommand{\Emod}{\desc{E}{mod}}

\newcommand{\Gtmp}{\desc{G}{tmp}}
\newcommand{\Etmp}{\desc{E}{tmp}}

\newcommand{\NAsample}{\desc{\overline{N}}{sample}_A}
\newcommand{\NBsample}{\desc{N}{sample}_B}

\newcommand{\ansSplit}{\answer{Split}}

\newcommand{\edgesB}{\oper{edgesB}}
\newcommand{\nonEdgesA}{\oper{nonEdgesA}}

\newcommand{\yesinstance}{\Inline{\textsc{Yes}-instance}}
\newcommand{\noinstance}{\Inline{\textsc{No}-instance}}

\newcommand{\dsfont}[1]{\ensuremath{\mathsf{#1}}}
\newcommand{\dsplit}{\dsfont{DSplit}}
\newcommand{\promisenl}{\dsfont{PromiseNL}}
\newcommand{\promisens}{\dsfont{PromiseNS}}
\newcommand{\nsample}{\dsfont{Wrapper}}

\newcommand{\invariant}[1]{\Inline{(\textit{#1})}}
\newcommand{\iEdge}{\invariant{edges}}
\newcommand{\iUpd}{\invariant{updates}}
\newcommand{\iProm}{\invariant{promise}}

\newcommand{\sumset}{\ensuremath{\texttt{sum}}}

\newcommand{\claimsdegreetext}{One can compute the value of $|N|$ in time $\Ohtilde(k)$.}

\newcommand{\claimcolortext}{With probability at least $1 - n^{-d}$, the following event holds: For every $b \in N$, there is an index
$i \in [\gamma d \log n]$ such that $b$ is exposed by $\chi_i$.}

\newcommand{\claimretrievetext}{Let $i \in [\gamma d \log n]$ be an index and
$c \in [\ell]$ be a color such that there is a vertex $b$
exposed by $\chi_i$ with color $\chi_i(b) = c$.
Then, given $c$ and $i$, one can retrieve the identifier of $b$ in time $\Ohtilde(k)$.}

\newcommand{\claimsmallNtext}{There is $i \in [\log n]$ so that
    $\Pr\left(|N_{i, j}| \not\in (12\ell, 72\ell) \textrm{ for all } j \in [d \log n]\right) \leq n^{-d\ell}$.}

\newcommand{\claimintersectiontext}{It holds that $|A_s \cap B_t| \leq \Oh(\sqrt{r + k})$
and $|A_t \cap B_s| \leq \Oh(\sqrt{r + k})$.}

\newcommand{\claimsubtext}{Let $a_1, a_2 \in A$ be vertices such that $a_1a_2 \not\in E(G)$. Then, at least one of the following holds:
\begin{itemize}
    \item there is a subset $U \subseteq V(G)$ such that $\{a_1, a_2\} \subseteq U$ and $G[U]$ is isomorphic to $2K_2$ or $C_4$;
    \item at least one of the vertices $a_1, a_2$ has at most $3\sqrt{k}$ neighbors in $B$.
\end{itemize}
Moreover, in time $\Ohtilde(k^2 d^2)$ we can find out
which of the above two cases holds,
and return either a set $U$ in the first case,
or the neighborhood of $a_1$ or $a_2$ in $B$
in the second case.
The answer is correct with probability at least
$1 - \Oh(n^{-d})$.}

\begin{document}
\maketitle

\iflipics
\else % arxiv
\thispagestyle{empty}
\begin{textblock}{20}(-1.9, 8.2)
    \includegraphics[width=40px]{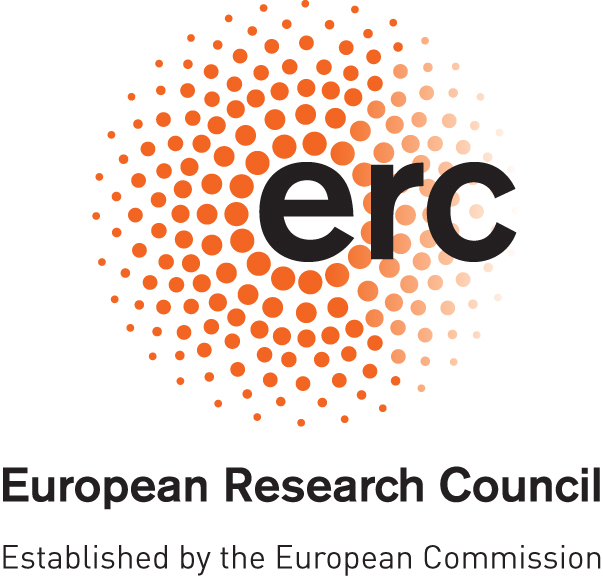}%
\end{textblock}
\begin{textblock}{20}(-2.15, 8.6)
    \includegraphics[width=60px]{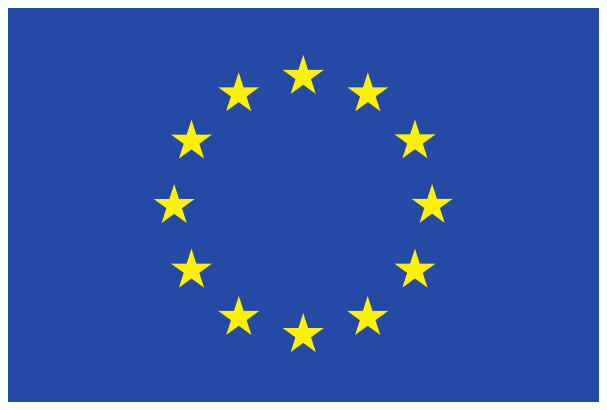}%
\end{textblock}
\fi

\begin{abstract}
We design a randomized data structure that, for a fully dynamic graph $G$ updated by edge insertions and deletions and integers $k, d$ fixed upon initialization, maintains the answer to the {\sc{Split Completion}} problem: whether one can add $k$ edges to $G$ to obtain a split graph. The data structure can be initialized on an edgeless $n$-vertex graph in time $n \cdot (k d \cdot \log n)^{\Oh(1)}$, and the amortized time complexity of an update is $5^k \cdot (k d \cdot \log n)^{\Oh(1)}$. The answer provided by the data structure is correct with probability $1-\Oh(n^{-d})$.
% \km{update times}

\end{abstract}

\section{Introduction}
\label{sec:introduction}

In the field of {\em{parameterized algorithms}}, one measures resources used by an algorithm not only in terms of the total size of the considered instance, but also in terms of auxiliary quantitative measures associated with the instance, called {\em{parameters}}. Recently there is a growing interest in applying this principle to the area of {\em{dynamic data structures}}. In this context, we typically consider an instance $I$ of a fixed problem of interest, with an associated parameter $k$. The instance is dynamic, in the sense that it is updated over time by problem-specific update operations, while for simplicity we assume that the parameter stays intact. The goal is to design a data structure that would maintain whether $I$ is a yes-instance of the problem under the updates to $I$. Here we allow the update time to depend in any computable way on the parameter and sublinearly on the instance size; for instance, we are insterested in (possibly amortized) update times of the form $f(k)$, $f(k)\cdot (\log n)^{\Oh(1)}$, or $f(k)\cdot n^{o(1)}$, where $f$ is some computable function and $n$ is the size of $I$. Note that we allow $f$ to be superpolynomial, so this framework may be applied even to $\mathsf{NP}$-hard problems, as long as their static parameterized variants are fixed-parameter tractable.

Parameterized dynamic data structures were first systematically investigated by Iwata and Oka~\cite{IwataO14}, followed by
Alman et al.~\cite{AlmanMW20}. These works provided data structures with update times $f(k)$ or $f(k)\cdot (\log n)^{\Oh(1)}$ for several classic problems such as {\sc{Vertex Cover}}, {\sc{Cluster Vertex Deletion}}, {\sc{Hitting Set}}, {\sc{Feedback Vertex Set}}, or {\sc{Longest Path}}. Other recent advances include data structures for maintaining various graph decompositions together with runs of dynamic programming procedures~\cite{ChenCDFHNPPSWZ21,DvorakKT14,DvorakT13,KorhonenMNPS23,MajewskiPS23}, treatment of parameterized string problems from the dynamic perspective~\cite{OlkowskiPRWZ23}, and even an application of the framework in the context of timed automata~\cite{GrezMPPR22}.

A topic within parameterized algorithms that could be particularly productive from the point of view of dynamic data structures is that of {\em{graph modification problems}}. In this context, we fix a graph class $\Cc$ and a set of graph operations $\Pi$ (e.g. vertex deletion, edge deletion, edge insertion), and consider the following parameterized problem: given a graph $G$ and a parameter $k$, decide whether one can apply at most $k$ operations from $\Pi$ to $G$ in order to obtain a graph belonging to $\Cc$. By instantiating different graph classes $\Cc$ and sets of operations $\Pi$, we obtain a wealth of parameterized problems with vastly different complexities, highly dependent on the combinatorics of the class $\Cc$ in question. The most widely studied are {\em{vertex deletion problems}} (only vertex deletions are allowed), {\em{edge deletion}} and {\em{completion problems}} (only edge deletions, respectively insertions, are allowed), and {\em{editing problems}} (both edge deletions and edge insertions are allowed).

Note that {\sc{Vertex Cover}}, {\sc{Feedback Vertex Set}}, and {\sc{Cluster Vertex Deletion}} can be understood as vertex deletion problems, for the classes of edgeless, acyclic, and cluster graphs, respectively. Thus, the results of Iwata and Oka~\cite{IwataO14} and of Alman et al.~\cite{AlmanMW20} already give parameterized dynamic data structures for some basic graph modification problems, but the impressive volume of work on static parameterized algorithms for such problems suggests that there is much more to be explored.

In this work we focus on the case of {\em{split graphs}}; recall that a graph $G$ is a split graph if the vertices of $G$ can be partitioned into sets $C$ and $I$ so that $C$ is a clique and $I$ is an independent set. As proved by F\"oldes and Hammer~\cite{FoldesH77}, split graphs are exactly graphs that exclude $2K_2$, $C_4$, and $C_5$ as induced subgraphs. Hence, a standard branching strategy solves {\sc{Split Vertex Deletion}} in time $5^k\cdot n^{\Oh(1)}$ and {\sc{Split Completion}} in time $5^k\cdot n^{\Oh(1)}$ (note that {\sc{Split Edge Deletion}} is equivalent to {\sc{Split Completion}} in the complement of the given graph). In fact, faster algorithms are known: {\sc{Split Vertex Deletion}} can be solved in time $1.2738^k\cdot k^{\Oh(\log k)}+n^{\Oh(1)}$~\cite{CyganP13}, while {\sc{Split Completion}} can be solved in subexponential parameterized time, more precisely in time $k^{\Oh(\sqrt{k})}\cdot n^{\Oh(1)}$, and admits a kernel with $\Oh(k^2)$ vertices~\cite{GhoshK0MPRR15}. Somewhat surprisingly, {\sc{Split Editing}} is polynomial-time solvable~\cite{TheSplittance}, while both {\sc{Split Vertex Deletion}} and {\sc{Split Completion}} remain $\mathsf{NP}$-hard~\cite{LewisY80,NatanzonSS01}. This makes graph modification problems related to split graphs very well understood in the static setting, and of remarkably low complexity, which suggests that they may serve as a suitable testbed for considerations from the point of view of dynamic data structures.

In fact, split graphs have already been considered in the dynamic setting. Ibarra~\cite{Ibarra08} gave a dynamic data structure with constant update time for the membership problem: the data structure maintains whether the graph in question is a split graph. This can be understood as the treatment of (any) modification problem for parameter $k$ equal to~$0$. Further aspects of dynamic maintenance of split graphs were investigated by Heggernes and Mancini~\cite{HeggernesM09}.

\subparagraph*{Our contribution.} In this work we propose a randomized dynamic data structure for the {\sc{Split Completion}} problem whose amortized update time depends polylogarithmically on the size of the graph. Formally, we prove the following result.

\newcommand{\mainthmtext}{There is a randomized data structure that for a fully dynamic graph $G$, updated by edge deletions and edge insertions, and a parameter $k$ fixed upon initialization, maintains the answer to the following question: can one add at most $k$ edges to $G$ to obtain a split graph. The data structure can be initialized on an edgeless $n$-vertex graph and an accuracy parameter $d\in \mathbb{N}$ in time $k^{\Oh(1)} \cdot d^2 \cdot n \cdot (\log n)^{\Oh(1)}$, and the amortized time complexity of updates is $5^k \cdot k^{\Oh(1)} \cdot d^2 \cdot (\log n)^{\Oh(1)}$. At all times, the answer provided by the data structure is correct with probability at least $1-\Oh(n^{-d})$.}

\begin{theorem}\label{thm:main}
\mainthmtext
\end{theorem}
\newtheorem*{mainthm}{\cref{thm:main}}

We note that by a tiny modification of the algorithm, we can also provide an analogous data structure for the dynamic {\sc{Split Edge Deletion}} problem, with same complexity~guarantees.

\subparagraph*{Overview.} Let us now discuss the main ideas behind the proof of \cref{thm:main}. The basic approach is to try to dynamize the standard $5^k\cdot n^{\Oh(1)}$-time branching algorithm, which works as follows: recursively find an obstruction --- an induced subgraph belonging to $\{2K_2,C_4,C_5\}$ --- and branch into at most $5$ different ways on how to destroy this obstruction by adding an edge. To implement this strategy in the dynamic setting, we need to have an efficient way of finding an obstruction in the current graph, or concluding that it is split. Indeed, if we get such a localization procedure, then the branching algorithm can be implemented using data structure's own methods for inserting and removing edges.
\iflipics
A detailed description is presented in \cref{sec:branching}.
\else
\fi

The next observation is that if $(G,k)$ is a \yesinstance of the {\sc{Split Completion}} problem, then in particular it is also a \yesinstance of the {\sc{Split Editing}} problem: one can obtain a split graph from $G$ by adding or removing at most $k$ edges. This is equivalent to the following statement: there is a partition $(A,B)$ of the vertex set of $G$ such that the number of edges with both endpoints in $B$ plus the number of non-edges with both endpoints in $A$ is at most $k$. For a partition $(A,B)$, this number is called the {\em{splittance}} of $(A,B)$. As observed by Hammer and Simeone~\cite{TheSplittance}, a partition with optimal splittance can be computed in polynomial time using a simple greedy argument. Ibarra~\cite{Ibarra08} dynamized this observation and showed that a partition $(A,B)$ with optimum splittance can be maintained with $\Ohtilde(1)$ worst-case update time. Moreover, upon every edge update, the data structure of Ibarra moves only $\Oh(1)$ vertices between the parts $A$ and $B$.

Therefore, thanks to the result of Ibarra we may assume that a partition $(A,B)$ with optimum splittance is available to us; this gives us a basic structural understanding of $G$. If the splittance of $(A,B)$ is larger than $k$, then we are certain that $(G,k)$ is a \noinstance of {\sc{Split Completion}}. Hence, from now on we will work in the {\em{promise model}}: We will assume that the maintained partition $(A,B)$ has splittance bounded by $k$ at all times. Lifting the result from the promise model to the general setting can be done by making a wrapper data structure that works as follows: while the splittance is larger than $k$, the new updates are put on a queue instead of directly implemented, and when the splittance becomes at most $k$ again, all the enqued updates are implemented in one large batch. We remark that this is the only source of amortization in our data structure: in the promise model, the claimed running time guarantees are worst-case.

The wrapper data structure is presented in \cref{sec:edges-inside}. We note that together with the partition $(A,B)$, the wrapper data structure also maintains sets $\nonEdgesA$ and $\edgesB$ consisting of non-edges with both endpoints in $A$ and edges with both endpoints in $B$, with a guarantee that they are correct whenever the splittance of $(A,B)$ is at most $k$. Note that then $|\nonEdgesA|+|\edgesB|\leq k$.

The next step is to implement efficient access to edges and non-edges crossing the partition $(A,B)$. More precisely, we need $\Ohtilde(k^{\Oh(1)})$-time queries that given $\ell\in \Oh(k)$, are able to list:
\begin{itemize}
 \item for a given $a\in A$, any set of $\ell$ neighbors of $a$ in $B$; and
 \item for a given $b\in B$, any set of $\ell$ non-neighbors of $b$ in $A$.
\end{itemize}
If $a$ has fewer than $\ell$ neighbors in $B$, then the query should list all of them; similarly for $b$. As the second query is symmetric to the first one, we focus on the latter.

\newcommand{\Ss}{\mathcal{S}}

First, consider the case when $a$ has exactly one neighbor in $B$. Then to quickly recover this neighbor, we can maintain in the data structure the following information:
\begin{itemize}
 \item the sum $s_A$ of the identifiers of all the vertices in $A$; and
 \item for every vertex $u$, the sum $s_u$ of the identifiers of all the neighbors of $u$.
\end{itemize}
Then given $a$, the identifier of the sole neighbor of $a$ in $B$ can be obtained by taking $s_a-s_A$ and adding the identifiers of all non-neighbors of $a$ in $A$, which can be listed in time $\Ohtilde(k)$ using the set $\nonEdgesA$. This trick can be lifted to listing all neighbors of $a$ in $B$ assuming that that their number is $\Oh(k)$ using the technique of {\em{color coding}} of Alon, Yuster, and Zwick~\cite{AlonYZ95}. Finally, to allow listing $\ell=\Oh(k)$ neighbors from possibly much larger neighborhoods, we sample in advance a polylogarithmically-sized family $\Ss$ of vertex subsets of varying sizes so that for any possible neighborhood $N$, with high probability among the sampled sets there will be some $S\in \Ss$ such that $\ell\leq |N\cap S|\leq \Oh(\ell)$. Then we apply the ideas presented above to list the members of $N\cap S$, whose number is already suitably bounded. We remark that this final trick is the only element of the reasoning that we do not know how to derandomize.
Implementation of this part of the proof is in \cref{sec:edges-across}.

The last step is to implement the procedure for finding of an obstruction. Here, the main idea is that every obstruction present in $G$ must contain a non-edge with both endpoints in $A$ or an edge with both endpoints in $B$; for otherwise it would be a split graph. For these (non-)edges, we have at most $k$ candidates contained in the lists $\nonEdgesA$ and $\edgesB$. Therefore, while searching for an obstruction, we can immediately ``anchor'' two of its vertices by guessing a member of $\nonEdgesA\cup \edgesB$ that is contained in the obstruction. Localizing the other two or three vertices of the obstruction requires a skillful juggle of the functionality specified above, in particular the methods for enumeration of (non-)neighbors across the partition $(A,B)$. This part of the argument can be found in \cref{sec:obstacles}.

\section{Preliminaries}
\label{sec:preliminaries}

For a positive integer $n$, we define $[n] \coloneqq \{1, 2, \ldots, n\}$.
If $x$ is a positive real number, we write $[x]$ as a shorthand for $[\ceil{x}] = \{ 1, 2, \ldots, \ceil{x} \}$.
For two sets $A$ and $B$, we denote their symmetric difference by
$A \symd B= (A \setminus B) \cup (B \setminus A)$.
Given any sets $A, B, C$ we write $C = A \uplus B$ if $C = A \cup B$, and $A \cap B = \emptyset$, that is, if $(A, B)$ is a partition of $C$.
For a boolean condition $b$, we denote by $\indic_b$ the indicator function of $b$, that is, $\indic_b = 1$ if $b$ is true, and $\indic_b = 0$ otherwise.

We use the $\Ohtilde(\cdot)$ notation to hide polylogarithmic factors; that is, $\Ohtilde(f(n))$ is the class of all functions $g(n)$ that are upper-bounded by $c\cdot f(n) \cdot (\log n)^c$ for some constant $c$.

\subparagraph*{Graphs.}
All graphs considered in this paper are undirected and simple, that is, contain no loops or parallel edges. We often assume that the vertex set of a graph on $n$ vertices is $[n]$.

Let $G$ be a graph.
We denote the sets of vertices and of edges of $G$ by $V(G)$ and $E(G)$, respectively.
We refer to an edge $\{u, v \} \in E(G)$ as $uv$, for every pair of vertices $u$ and $v$.
For a vertex $v \in V(G)$, we denote the degree of $v$ by $d_G(v)$
and the set of all neighbors of $v$ by $N_G(v)$
(we may skip index $G$ if it is clear from the context).
Additionally, we define $\Non_G(v)$ as the set of non-neighbors of $v$, that is, $\Non_G(v) = V(G) \setminus (N_G(v) \cup \{ v \})$.

For a subset $A \subseteq V(G)$, we denote by $G[A]$ the subgraph of $G$ induced by $A$.
For a graph $F$, we say that $G$ is $F$-free if it does not contain an induced subgraph isomorphic to~$F$.
Similarly, for a family of graphs $\Fc$, a graph $G$ is $\Fc$-free if it is $F$-free for every $F \in \Fc$.

\subparagraph*{Split graphs.}
Let $G$ be a graph.
We say that $G$ is \emph{split} if its vertices can be partitioned into two sets $A$ and $B$ such that $G[A]$ is a clique, and $G[B]$ is an independent set.
For any partition $(A, B)$ of $V(G)$, we denote
\[
    \splittance_G(A, B) \coloneqq \binom{|A|}{2} - |E(G[A])| + |E(G[B])|,
\]
that is, $\splittance_G(A, B)$ counts the number of non-edges of $G$ within $A$ and edges of $G$ within $B$.
% Note that $(A, B)$ is a split partition of $V(G)$ if and only if $\splittance_G(A, B) = 0$.
For the whole graph $G$, we write
\[
    \splittance(G) \coloneqq \min_{(A, B) \colon V(G) = A \uplus B} \splittance_G(A, B).
\]
One can observe that $\splittance(G)$ counts the minimal number of edge updates (insertions or deletions) that need to be done on $G$ in order to make the graph $G$ split.
In particular, $G$ is split if and only if $\splittance(G) = 0$. It turns out that a partition with optimum splittance can be computed greedily by taking a prefix of the vertex set ordered by decreasing degrees, as stated formally in the following observation of Hammer and Simeone~\cite{TheSplittance}.

\begin{fact}[\cite{TheSplittance}]
    \label{fact:split-degrees}
    Let $G$ be a graph on $n$ vertices, and let $d_1 \geq d_2 \geq \ldots \geq d_n$ be a sorted sequence of the vertex degrees of $G$.
    Define $m \coloneqq \max \{ i \in [n] \mid d_i \geq i - 1 \}$.
    Then,
    \[
        \splittance(G) = m (m - 1) - \sum_{i = 1}^m d_i + \sum_{i = m + 1}^n d_i = 0,
    \]
    and a certifying partition $(A, B)$ of $V(G)$ can be obtained by taking the vertices of $G$ corresponding to the degrees $(d_1, \ldots, d_m)$ and $(d_{m+1}, \ldots, d_n)$, respectively.
\end{fact}

We will also make use of another characterization of split graphs, namely via a finite family of forbidden induced subgraphs.

\begin{fact}[\cite{FoldesH77}]
    \label{fact:split-forbidden}
    A graph $G$ is split if and only if $G$ is $\Fc$-free, where $\Fc = \{2K_2, C_4, C_5\}$.
\end{fact}

\subparagraph*{Computation model and basic data structures.} We assume the standard word RAM model of computation with machine words of length $\Oh(\log n)$. As the most basic data structures, we use {\em{sets}} (representing a set of objects that can be updated by insertions and deletions, and queried for membership) and {\em{dictionaries}} (representing a set of key-value pairs that can be updated by insertions and deletions, and queried for the value associated with a given key). For both these data structures, we use standard implementations using any kind of self-balancing binary search trees, which offer all operations in worst-case time $\Oh(\log n)\subseteq \Ohtilde(1)$. In particular, graphs are represented as dictionaries mapping vertices to the sets of their neighbors. We refer to this representation as \emph{adjacency list representation}. It allows verifying the adjacency of a pair of vertices in time $\Ohtilde(1)$.

Most of our data structures are randomized. We would like to stress that the random choices happen only during the initialization, and therefore the events of incorrectness of further queries are {\em{not}} independent --- they all depend on the initial randomness.

\section{Dynamic splittance}
\label{sec:partition}

In this section, we discuss how to maintain the value of $\splittance(G)$ with a witnessing partition $(A, B)$ of $V(G)$ for a dynamic graph $G$.
The data structure is a simple adaptation of the data structure of Ibarra \cite{Ibarra08} who showed how to maintain whether a dynamic graph $G$ is split.
\iflipics
For the sake of self-containedness, the proof is presented in \cref{sec:app-splittance}.
\else
We include the proof for the sake of self-containedness.
\fi

\newcommand{\lemsplittext}{%
There exists a data structure $\dsplit[n]$ that runs on a dynamic graph $G = (V, E)$ on $n$ vertices, and maintains a partition $(A, B)$ of its vertices satisfying $\splittance_G(A, B) = \splittance(G)$.
The data structure supports the following operations:
\begin{itemize}
    \item $\oper{initialize}(n)$: fixes the set of vertices
    $V \coloneqq [n]$ for the entire run, sets the initial graph $G$
    to an edgeless graph $G \coloneqq (V, \emptyset)$ and
    sets the initial partition $(A,B)$ to
    $(A, B) \coloneqq (\emptyset, V)$. Runs in time $\Ohtilde(n)$.
    \item $\oper{update}(uv)$: inserts edge $uv$
    if $uv \not\in E(G)$ or removes edge $uv$ if $uv \in E(G)$.
    Let $G' \coloneqq (V, E(G) \symd \{ uv \})$ be the updated graph
    and $(A',B')$ be the updated partition for $G'$.
    The method returns a set $\Vmoved \subseteq V$ of size
    $|\Vmoved| \leq \Oh(1)$
    such that
    $
        (A', B') = (A \symd \Vmoved, B \symd \Vmoved)
    $.
    %of vertices of $G$ satisfies $\splittance_{G'}(A', B') = \splittance(G')$, and it is guaranteed that $|\Vmoved| \leq \Oh(1)$.
    The running time is $\Ohtilde(1)$.
    \item{$\oper{splittance}()$: returns the current value of $\splittance_{G}(A, B) = \splittance(G)$ in time~$\Oh(1)$.}
\end{itemize}
}

\begin{lemma}\label{lem:splittance}
    \lemsplittext
\end{lemma}
\newtheorem*{lemsplit}{Lemma~\ref{lem:splittance}}

\iflipics

\else % arxiv
\begin{proof}
The algorithm maintains the following data:
\begin{itemize}
    \item the standard representation of the dynamic graph $G$ (i.e., for each vertex $v \in V(G)$, we store a dictionary of all the vertices adjacent to~$v$);
    \item a sorted list $\Lc$ of all the vertex degrees of $G$: $d_1 \geq d_2 \geq \ldots \geq d_n$;
    \item an index $m \in [n]$ equal to $\max \{ i \in [n] \mid d_i \geq i - 1 \}$ and a pointer to an element of $\Lc$ corresponding to $d_m$;
    \item the value of $\splittance(G)$ computed according to \cref{fact:split-degrees};
    \item the sets $A$ and $B$ partitioning $V(G)$ defined as in \cref{fact:split-degrees};
    \item a dictionary $\oper{degree}$ mapping each vertex $v$ of $G$ a pointer to an element of $\Lc$ equal to~$d_G(v)$, where all the pointers are assumed to be pairwise different, and a dictionary $\oper{vertex}\colon \Lc \to V(G)$ storing the inverse mapping $\oper{degree}^{-1}$;
    \item dictionaries $\oper{firstD}, \oper{lastD}$ which, for every $i \in [n]$, store the pointers to the first and the last element of $\Lc$ equal to $i$, respectively (or, null pointers if there are no such elements in $\Lc$).
\end{itemize}
Clearly, given an integer $n$, all the data above can be initialized in time $\Ohtilde(n)$.

Now, consider an update of an edge $e = \{u, v\}$.
Having the graph $G$, we can verify whether $e$ should be inserted or removed in time $\Oh(1)$.
We can easily update the graph $G$ to $G' \coloneqq (V, E \symd \{ e \})$, and update the list $\Lc' \coloneqq (d_1' \geq d_2' \geq \ldots \geq d_n')$ with the use of dictionaries $\oper{degree}$, $\oper{firstD}$ and $\oper{lastD}$.
All these computations can be done in time $\Ohtilde(1)$.

Next, we need to compute the value of $m' \coloneqq \max \{ i \in [n] \mid d'_i \geq i - 1 \}$.
Ibarra (\cite[Section~7]{Ibarra08}) observed that after a single edge update, it holds that $|m' - m| \leq \Oh(1)$.
Hence, we can find the desired index $m'$ in time $\Oh(1)$, by inspecting the $\Oh(1)$-sized neighborhood on $\Lc$ of the list element corresponding to $d_m$.
Consequently, using the formulas from \cref{fact:split-degrees} the value of $\splittance(G')$ can be computed from $\splittance(G)$ in time $\Oh(1)$, and only a constant number of vertices need to be moved between $A$ and $B$ in the new partition $(A', B')$.
Again, it is easy to verify that all the remaining data can be updated in time $\Ohtilde(1)$, and the correctness is guaranteed by~\cref{fact:split-degrees}.
\end{proof}

\fi

\section{Listing (non-)edges across the partition}
\label{sec:edges-across}

Let $G = (A \uplus B, E)$ be a dynamic graph whose vertices are partitioned into two sets $A$ and $B$ which might change over time.
In this section, we devise two auxiliary data structures that are responsible for listing all neighbors of a given $a \in A$ in the set $B$, provided this neighborhood is small, or sampling a sufficient number of neighbors of $a$ in $B$.
Symmetric queries can be also given for non-neighborhoods in $A$ of vertices $b \in B$.

Both data structures work under the promise that $\splittance_G(A, B) \leq k$. This promise will be lifted in the next section, where we design a wrapper data structure that works also when $\splittance_G(A, B) > k$. To facilitate the design of this wrapper data structure, both data structures presented in this section need to be able to process updates in batches: instead of being given only single edge updates, they will be given sets of edge modifications such that after applying all of them the condition $\splittance_G(A, B) \leq k$ is preserved.

Also, both data structures are assumed to have access to sets $\nonEdgesA, \edgesB$ consisting of non-edges with both endpoints in $A$ and edges with both endpoints in $B$, respectively. This access will be provided by the wrapper data structure.

We now give the first data structure, called $\promisenl$ (standing for neighbors listing).

\begin{lemma}\label{lem:listing1}
    There is a data structure $\promisenl[n, k, \ell, d, S]$,
    parametrized by integers $k, \ell, d \in \N$ and a subset
    $S \subseteq V(G)$,
    that runs on a dynamic graph $G = (V, E)$ with $V=[n]$,
     and maintains a partition $(A, B)$ of vertices of $G$
     %satisfying $\splittance_G(A, B) = \splittance(G)$
     by supporting
     the following operations:
    \begin{itemize}
        \item $\oper{initialize}(n, k, \ell, d, S)$: sets
        $V \coloneqq [n]$, fixes the values of $k$, $\ell$, $d$
        and $S$ for the entire run, initializes
        $G \coloneqq (V, \emptyset)$ and
        $(A, B) \coloneqq (\emptyset, V)$.
        Runs in time $\Ohtilde(\ell d \cdot n)$.

        \item $\oper{update}(\Vmoved, \Emod, \nonEdgesA, \edgesB)$:
        modifies the edge set of the underlying graph to
        $E' \coloneqq E \symd \Emod$ and modifes the partition to
        $(A',B') \coloneqq (A \symd \Vmoved,B \symd \Vmoved)$, provided that
        after the update $\splittance_{G'}(A', B') \leq k$ for the
        modified graph $G'=(V,E')$.
        Requires two lists of edges
        $\nonEdgesA, \edgesB \subseteq \binom{V}{2}$ storing the sets of
        all non-edges of $G'[A']$ and all edges of $G'[B']$, respectively.
        Note that both sets are of size at most $k$ since
        $\splittance_{G'}(A', B') \leq k$.
        The running time is
        $\Ohtilde(k + \ell d \cdot (|\Vmoved| + |\Emod|))$.
%
        %$G' \coloneqq (A' \uplus B', E')$, where
        %\[
        %    A' \coloneqq A \symd \Vmoved, \qquad B' \coloneqq B \symd \Vmoved, \qquad E' = E \symd \Emod.
        %\]
        %Is it guaranteed that the graph~$G'$ satisfies
        %$\splittance_{G'}(A', B') \leq k$.
%
    \end{itemize}
    Moreover, the data structure can answer the following queries:
    \begin{itemize}
        \item $\oper{listNeighborsBS}(a)$: given a vertex $a \in A$, either returns the set
        \[
            N \coloneqq  N_G(a) \cap B \cap S
        \]
        provided $|N| \leq \ell$, or reports \textsc{TooMany} if $|N| > \ell$.

        \item $\oper{listNonNeighborsAS}(b)$: given a vertex $b \in B$, either returns the set
        \[
            N \coloneqq  \Non_G(b) \cap A \cap S
        \]
        provided $|N| \leq \ell$, or reports \textsc{TooMany} if $|N| > \ell$.
    \end{itemize}
    Both queries run in time $\Ohtilde(k \ell d)$ and return correct answers with probability $1 - \Oh(n^{-d})$.
\end{lemma}

\begin{proof}
    Let $\gamma$ be some large enough constant whose value will be set later.
    We assume that $\ell \leq n$, for otherwise we might set $\ell \coloneqq  n$, and this would not affect answers to the queries.
    We also denote the sum of the vertex indices in a set
    $X \subseteq V$ as $\sumset(X)=\sum_{v \in X} v$.
    The data structure $\promisenl[n, k, \ell, d, S]$ stores:
    \begin{itemize}
        \item the adjacency list representation of $G$ and
        representation of $A$ and $B$ as sets;%\anka{See preliminaries}
        \item the integer values of $k, \ell, d$ and $S \subseteq V(G)$ represented as a set;%\anka{see preliminaries}
        \item the sets of edges
         $\nonEdgesA, \edgesB \subseteq \binom{V(G)}{2}$
         provided in the updates, as lists;
     %   \item a bijective mapping $\oper{id}\colon V(G) \to [|V(G)|]$ and its inverse mapping $\oper{id}^{-1} \colon [|V(G)|] \to V(G)$;
        \item for every $i \in [\gamma d \log n]$, a vertex $\ell$-coloring $\chi_i \colon V(G) \to [\ell]$.
        Define $S_{i, c} \coloneqq  S \cap \chi_i^{-1}(c)$, for every $i \in [\gamma d \log n]$ and $c \in [\ell]$; these sets are also stored in the data structure;
        \item the following integer values:
        \begin{align*}
            \oper{countS}(v) & = |N_G(v) \cap S|\quad \text{for every } v \in V(G);\\
            \oper{countAS} & = |A \cap S|; \qquad \oper{countBS} = |B \cap S|; \\
            \oper{idSumS}_{i, c}(v) & = \sumset( N_G(v) \cap S_{i, c} )\quad \text{for every } i \in [\gamma d \log n], c \in [\ell] \text{ and } v \in V(G); \\
            \oper{idSumAS}_{i,c} & = \sumset( A \cap S_{i, c} )\quad \text{for every } i \in [\gamma d \log n] \text{ and } c \in [\ell];\\
            \oper{idSumBS}_{i,c} & = \sumset( B \cap S_{i, c} ) \quad \text{for every } i \in [\gamma d \log n] \text{ and } c \in [\ell].
        \end{align*}
    \end{itemize}

    Before proceeding with the description of operations, let us give a brief idea how these variables are going to be used.
    Consider a vertex $a \in A$.
    Using the variables $\oper{count}$- and the set $\nonEdgesA$, one can obtain the number of neighbors of $a$ in the set $B \cap S$ by applying a simple inclusion-exclusion formula.
    If this neighborhood $N$ is  of size at most~$\ell$, we would like to list all the vertices of $N$.
    To this end, we use the vertex colorings $\chi_i$ and the corresponding variables $\oper{idSum}$-.
    Each coloring $\chi_i$ will be initialized randomly.
    Therefore, with high probability, for every vertex $v \in N$, there will be a coloring $\chi_i$ such that under~$\chi_i$, $v$ has a different color than all the other vertices of $N \setminus \{v\}$.
    By applying a similar inclusion-exclusion principle on the variables $\oper{idSum}$- we will be able to retrieve an identifier of $v$.
    
    \subparagraph*{Initialization.}
    \begin{itemize}
        \item We initialize the adjacency representation of $G$
            to store an edgless graph on vertices~$[n]$.
            We initialize $A=\emptyset$ and $B=V$.
            Also, we set the lists
            $\nonEdgesA,\edgesB$ to be empty.
        \item We store the values of $k, \ell, d$ and $S$ as a set.
        %\item The bijection $\oper{id}$ is chosen arbitrarily and
        %is fixed for the entire run.
        \item All colorings $\chi_i$ are initialized randomly, that is, for every $i \in [\gamma d \log n]$ and $v \in V(G)$, we choose a color for $\chi_i(v)$ uniformly at random from the set $[\ell]$.
        The colorings are fixed for the entire run as well, and so are sets $S_{i,c}$.
        \item Since the initial graph is empty, all the values $\oper{count}$- and $\oper{idSum}$- can be computed according to their definitions in total time $\Ohtilde(\ell d \cdot n)$.
    \end{itemize}
    
    \subparagraph*{Update.}
    Consider an update $\oper{update}(\Vmoved, \Emod, \nonEdgesA, \edgesB)$.
    First, we save the new values of $\nonEdgesA$ and $\edgesB$.
    Observe that the functions $\oper{countS}(\cdot)$ and $\oper{idSumS}_{i, c}(\cdot)$ need to be changed only for the endpoints of the edges of $\Emod$, and the remaining $\Ohtilde(\ell  d)$ stored variables are affected only by vertices moved between $A$ and $B$.
    Therefore, the entire update takes $\Ohtilde(k + \ell d \cdot (|\Vmoved| + |\Emod|))$ time.
    
    \subparagraph*{Queries.}
    %Fix a vertex $a \in A$.
    We show how to answer the query $\oper{listNeighborsBS}(a,\ \nonEdgesA)$
    for a vertex $a \in A$.
    Recall that since $\splittance_G(A, B) \leq k$, there are at most $k$ non-edges within $A$, i.e., $|\nonEdgesA| \leq k$.
    Let $N \coloneqq  N_G(a) \cap B \cap S$.
    We start with the following observation.
    
    \begin{claim}\label{claim:sdegree}
        \claimsdegreetext % -> macros-konrad
    \end{claim}
    \begin{claimproof}
    The following equalities hold:
\begin{align*}
    |N| = |N_G(a) \cap B \cap S| & = |N_G(a) \cap S| - |N_G(a) \cap A \cap S| = \oper{countS}(a) - |N_G(a) \cap A \cap S| \\
    & = \oper{countS}(a) - (|A \cap S| - \indic_{a \in S} - |\Non_G(a) \cap A \cap S|) \\
    & = \oper{countS}(a) - (\oper{countAS} - \indic_{a \in S} - |\Non_G(a) \cap A \cap S|).
\end{align*}
Observe that $\Non_G(a) \cap A \cap S \subseteq \Non_G(a) \cap A$, and we can obtain all the non-neighborhood of $a$ in $A$ by iterating over the given set $\nonEdgesA$ of size at most $k$.
Thus, we can compute the value of $|\Non_G(a) \cap A \cap S|$ in time $\Ohtilde(k)$.
    \end{claimproof}

    Now, as a first step to answer
    $\oper{listNeighborsBS}(a,\ \nonEdgesA)$ query, we compute $|N|$.
    If $|N| > \ell$, we report \textsc{TooMany}.
    So assume that $|N| \leq \ell$.
    For a vertex $b \in N$ and a coloring $\chi_i$,
    let us say that $b$ is \emph{exposed} by the coloring $\chi_i$
    if $\chi_i(b) \neq \chi_i(b')$ for every $b' \in N \setminus \{b\}$. First, we note that every member of $N$ is exposed with high probability.
    %(Proofs of statements marked with ($\app$) can be found in the Appendix.)

    \iflipics
    \begin{claim}[$\app$]\label{claim:color}
    \else
    \begin{claim}\label{claim:color}
    \fi
        \claimcolortext % -> macros-konrad
    \end{claim}
    \newtheorem*{claimcolor}{\cref{claim:color}}
    
    \iflipics
    \else
    \begin{claimproof}
    %Denote the event above by $Z$.
If $\ell = 1$, the claim is trivial, so assume that $\ell \geq 2$.
After fixing a vertex $b \in N$ and a coloring $\chi_i$, we obtain that
\[
\Pr\left( \text{$b$ is exposed by $\chi_i$} \right) = \frac{\ell \cdot (\ell - 1)^{|N|-1}}{\ell^{|N|}} = \left( \frac{\ell - 1}{\ell} \right)^{|N|-1} \geq  \left( \frac{\ell - 1}{\ell} \right)^{\ell-1} > \frac{1}{10e}.
\]
Therefore, by applying the union bound we can upper bound the
probability of the event complementary to the one in the claim as follows:
\begin{align*}
&\Pr \left( \exists_{b \in N}\  \forall_{i \in [\gamma d \log n]}\ \text{ $b$ is not exposed by $\chi_i$} \right)
\leq \sum_{b \in N} \Pr \left(\forall_{i \in [\gamma d \log n]}\ \text{ $b$ is not exposed by $\chi_i$} \right) \\
& \leq \sum_{b \in N} \left( 1 - \frac{1}{10e} \right)^{\gamma d \log n}
\leq n \cdot \left( 1 - \frac{1}{10e} \right)^{\gamma d \log n}
\leq \frac{1}{n^d},
\end{align*}
where the last inequality holds for all large enough values of $\gamma$.
    \end{claimproof}
    \fi
    
    Next, we note that once a vertex of $N$ is exposed, its identifier can be easily computed using a similar inclusion-exclusion principle as in the proof of \cref{claim:sdegree}.

    \iflipics
    \begin{claim}[$\app$]\label{claim:retrieve}
    \else
    \begin{claim}\label{claim:retrieve}
    \fi
        \claimretrievetext % -> macros-konrad
    \end{claim}
    \newtheorem*{claimretrieve}{\cref{claim:retrieve}}
    
    \iflipics
    \else
    \begin{claimproof}
    Since $b$ is exposed by $\chi_i$, we obtain that
\begin{align*}
    b = \sumset( N \cap \chi^{-1}(c) ) =
    \sumset( N_G(a) \cap B \cap S \cap \chi^{-1}(c) ) =
    \sumset( N_G(a) \cap B \cap S_{i,c} ).
\end{align*}
Now, we apply a similar inclusion-exclusion formula as in the proof of \cref{claim:sdegree}:
\begin{align*}
    \sumset( N_G(a) &\cap B \cap S_{i,c} ) = \sumset( N_G(a) \cap S_{i,c} ) - \sumset( N_G(a) \cap A \cap S_{i,c} ) \\
    & = \oper{idSumS}_{i,c}(a) - \sumset( N_G(a) \cap A \cap S_{i,c} ) \\
    & = \oper{idSumS}_{i,c}(a) - \left( \sumset( A \cap S_{i,c}) - a \cdot \indic_{a \in S_{i,c}}  - \sumset( \Non_G(a) \cap A \cap S_{i,c} )  \right) \\
    & = \oper{idSumS}_{i,c}(a) - \left( \oper{idSumAS}_{i,c} - a \cdot \indic_{a \in S_{i,c}} - \sumset( \Non_G(a) \cap A \cap S_{i,c} ) \right).
\end{align*}
Again, we observe that $\Non_G(a) \cap A \cap S_{i,c} \subseteq \Non_G(a) \cap A$, and the latter set can be easily enumerated using the set $\nonEdgesA$ in time $\Ohtilde(k)$.
    \end{claimproof}
    \fi
    
    We already have all the tools to finish the implementation of the query $\oper{listNeighborsBS}(a)$.
    To list all the elements of $N$,
    it is enough to iterate over all indices
    $i \in [\gamma d \log n]$ and colors $c \in [\ell]$,
    and for each pair $(i, c)$ retrieve the
    value that is described in \cref{claim:retrieve}
    and check whether the returned value is a vertex that belongs to
    $N$.
    By \cref{claim:color}, this procedure will retrieve the entire
    set $N$ with probability at least $1 - n^{-d}$ and the total
    running time is bounded by $\Ohtilde(k \ell d)$.
    
    To answer the queries $\oper{listNonNeighborsAS}(b,\ \edgesB)$ for $b \in B$ we proceed analogously.
    To this end, we need additionally the values of:
    \begin{itemize}
        \item $|\Non_G(b) \cap S|$ which can be obtained from the values of $|N_G(b) \cap S|$ and $|S|$; and
        \item $\sumset( \Non_G(b) \cap S_{i, c} ) $ which we can
        obtain from $\sumset( N_G(b) \cap S_{i, c} )$ and
        $\sumset( S_{i, c} ) $, where the latter sums can be precomputed during the initialization.\qedhere
    \end{itemize}
\end{proof}

We note that the application of color coding presented in \cref{claim:color} and \cref{claim:retrieve} can be derandomized in a standard way, by replacing sampling the sets $S_{i,c}$ with drawing them from a suitable universal hash family; see e.g.~\cite[Section~5.6]{platypus}. However, as we will use randomization again in our second data structure (\cref{lem:sampling} below), we prefer to present color coding here in a randomized, and therefore simpler manner.

So far, we can list $\ell$ neighbors of a vertex $a \in A$ in the set $B \cap S$ only under an assumption that this neighborhood $N$ is small.
However, this condition does not need to hold in general, and in
the next sections we will also need to find several neighbors of $a$ in $B$ (more precisely, a subset of $N$ of size $\Theta(k)$) in the case when the set $N$ is large.
This is the responsibility of our next data structure, named $\promisens$ (standing for neighbors sampling).
Roughly speaking, this data structure will store a set of layers
of the vertex set, where the $i$th layer contains
roughly $n / 2^i$ random vertices of $G$.
Then, provided $|N| > \ell$,
with high probability one of the layers will contain $\Theta(\ell)$ vertices of $N$. By running an instance of $\promisenl$ on each layer, we will be able to retrieve those vertices.

\begin{lemma}\label{lem:sampling}
    There exists a data structure $\promisens[n, k, \ell, d]$,
    parametrized by integers $k, \ell, d \in \N$,
    that runs on a dynamic graph $G = (V, E)$ with $V=[n]$,
     and maintains a dynamic partition $(A, B)$ of vertices of $G$
     %satisfying $\splittance_G(A, B) = \splittance(G)$
     by supporting
     the following operations:
    \begin{itemize}
        \item $\oper{initialize}(n, k, \ell, d)$: sets
        $V \coloneqq [n]$, fixes the values of $k$, $\ell$ and $d$
        for the entire run, initializes
        $G \coloneqq (V, \emptyset)$ and
        $(A, B) \coloneqq (\emptyset, V)$.
        Runs in time $\Ohtilde(\ell d^2 \cdot n)$.

        \item $\oper{update}(\Vmoved, \Emod, \nonEdgesA, \edgesB)$:
        modifies the edge set of the underlying graph to
        $E' \coloneqq E \symd \Emod$ and modifes the partition to
        $(A',B') \coloneqq (A \symd \Vmoved,B \symd \Vmoved)$, provided that
        after the update $\splittance_{G'}(A', B') \leq k$ for the
        modified graph $G'=(V,E')$.
        Requires two lists of edges
        $\nonEdgesA, \edgesB \subseteq \binom{V}{2}$ storing the sets of
        all non-edges of $G'[A']$ and all edges of $G'[B']$, respectively.
        Note that both sets are of size at most $k$ since
        $\splittance_{G'}(A', B') \leq k$.
        The running time is
        $\Ohtilde(kd + \ell d^2 \cdot (|\Vmoved| + |\Emod|))$.

    \end{itemize}
    Moreover, the data structure can answer the following queries:
    \begin{itemize}
        \item $\oper{sampleEdges}(a)$: given a vertex $a \in A$, returns a subset
        \[
            \NBsample(a) \subseteq N_G(a) \cap B
        \]
        of size $\min(\ell, |N_G(a) \cap B|)$;
        \item $\oper{sampleNonEdges}(b)$: given a vertex $b \in B$, returns a subset
        \[
            \NAsample(b) \subseteq \Non_G(b) \cap A
        \]
        of size $\min(\ell, |\Non_G(b) \cap A|)$.
    \end{itemize}
    Both queries run in time $\Ohtilde(k\ell d^2)$ and
    return correct answers with probability
    $1 - \Oh(n^{-d})$.
\end{lemma}

\begin{proof}
    The data structure $\promisens[n, k, \ell, d]$ stores:
    \begin{itemize}
        \item the adjacency lists representation of $G$
        and representation of $A$ and $B$ as sets;
        \item the values of the parameters $k$, $\ell$, and $d$;
        \item the sets
          $\nonEdgesA, \edgesB \subseteq \binom{V(G)}{2}$
          given in the updates, as lists;
        \item a single instance $\promisenl[n, k, \ell', d, S]$ of the data structure from \cref{lem:listing1} running on the same dynamic graph~$G$, where $\ell' \coloneqq 72\ell$ and $S \coloneqq V(G)$;
        \item for every $i \in [\log n]$ and $j \in [d \log n]$, a subset of vertices $V_{i, j} \subseteq V(G)$, fixed upon initialization and not modified later on; and
        \item for every $i \in [\log n]$ and $j \in [d \log n]$, an instance $\promisenl_{i,j}[n, k, \ell', d, S_{i, j}]$ of the data structure from \cref{lem:listing1} running on the same dynamic graph~$G$, where $S_{i,j} \coloneqq V_{i, j}$.
    \end{itemize}
    % Note that whenever we call $\oper{listNeighborsBS}(a, \nonEdgesA)$ or $\oper{listNonNeighborsAS}(b, \edgesB)$ on the structures from \cref{lem:listing1}, for some vertices $a \in A$, $b \in B$, we need to know the values of the sets $\nonEdgesA$ and $\edgesB$ before.
    % To this end, we first call $\oper{listNonEdgesA}$ and $\oper{listEdgesB}$ on $\Dinner$, and the running time is dominated by the calls to the structures from \cref{lem:listing1}.
    % In what follows, we omit these calls to the structure $\Dinner$ for simplicity. 

    \subparagraph*{Initialization.}
    Intuitively, we want the set $V_{i, j}$ to be a random subset containing roughly $n/2^i$ vertices of $G$.
    To achieve this, for every $i \in [\log n]$ and $j \in [d \log n]$, and every vertex $v \in V(G)$, we put $v$ into the set $V_{i, j}$ with probability $2^{-i}$, independently.
    Then, we initialize all the data structures $\promisenl$ and $\promisenl_{i,j}$.
    The total running time of this phase is bounded by
    \[ \log n \cdot d \log n \cdot \Ohtilde(\ell' d \cdot n) \subseteq \Ohtilde(\ell d^2 \cdot n).\]

    \subparagraph*{Update.}
    Given an update with arguments $(\Vmoved, \Emod, \nonEdgesA, \edgesB)$, we simply pass it to all the stored data structures.
    The running time is thus bounded by
    \[
        \log n \cdot d \log n \cdot \Ohtilde(k + \ell' d \cdot (|\Vmoved| + |\Emod|)) \subseteq \Ohtilde(kd + \ell d^2 \cdot (|\Vmoved| + |\Emod|)).
    \]

    \subparagraph*{Sampling.}
    Fix a vertex $a \in A$.
    We show how to answer the query $\oper{sampleEdges}(a)$.
    The implementation of queries $\oper{sampleNonEdges}(b)$, for $b \in B$, is symmetric.
    Let $N \coloneqq N_G(a) \cap B$.
    Our goal is to sample a subset of the set $N$ of size $\min(\ell, |N|)$.

    First, consider the case when $|N| \leq \ell'$.
    Then, by querying the structure $\promisenl$ about the list of
    neighbors of $a$ in the set $B \cap V(G) = B$
    (i.e., we run $\oper{listNeighborsBS}(a)$), we obtain all the
    vertices of $N$ with probability at least $1 - \Oh(n^{-d})$.
    Otherwise, the call returns $\textsc{TooMany}$, and we can proceed as follows.

    From now on, assume that $|N| > \ell'=72\ell$.
    Then, we query all the data structures $\promisenl_{i, j}$ about the list of neighbors of $a$ in the corresponding sets $B \cap V_{i, j}$ (i.e., we run $\oper{listNeighborsBS}(a)$ on each $\promisenl_{i, j}$).
    We prove that with probability
    at least $1 - \Oh(n^{-d})$ one of such queries will correctly return at least $\ell$ elements of $N$.
    Define $N_{i, j} \coloneqq N \cap V_{i, j}$, for every $i \in [\log n]$, $j \in [d \log n]$.

    \iflipics
    \begin{claim}[$\app$]\label{claim:smallN}
    \else
    \begin{claim}\label{claim:smallN}
    \fi
        \claimsmallNtext % -> macros-konrad
    \end{claim}
    \newtheorem*{claimsmallN}{\cref{claim:smallN}}
    
    \iflipics
    \else
    \begin{claimproof}
    Recall that $|N| > 72\ell$ and each set of the form $N, N_{1, j}$, $N_{2, j}$, $N_{3, j}, \ldots,$ (for every $j \in [d \log n]$) contains half as many vertices in expectation as the previous one. Hence, there exists an index $i \in [\log n]$ such that
\[
    \E |N_{i, j}| \in [24\ell, 48\ell] \quad \text{ for every } j \in [d \log n].
\]
Fix such an index $i$ and let $\mu \coloneqq \E |N_{i, j}|=n/2^i$.
Observe that for every $j$, $|N_{i, j}|$ is the sum of $n$
random variables $X_v$ taking values $0$ or $1$, indicating
whether a vertex $v$ belongs to $N_G(a) \cap B \cap V_{i,j}$.
Hence, we can apply the following Chernoff's inequality:
\[
    \Pr \left( |N_{i, j}| \not\in \left(\frac{\mu}{2}, \frac{3\mu}{2} \right) \right) \leq 2e^{-\mu/12} \quad \text{ for every } j \in [d \log n].
\]
By plugging in the bounds on $\mu$ (i.e., $\mu \in [24\ell, 48\ell]$), we obtain that
\[
    \Pr \left(|N_{i, j}| \not\in (12\ell, 72\ell) \right) \leq 2e^{-24\ell / 12} < 2 \cdot 2^{-2\ell} < 2^{-\ell} \quad \text{ for every } j \in [d \log n].
\]
Finally, since the random variables $|N_{i, j}|$ are
independent over the choices of $j \in [d \log n]$, it holds that
\[
    \Pr \left(|N_{i, j}| \not\in (12\ell, 72\ell) \text{ for every } j \in [d \log n] \right) \leq 2^{-\ell \cdot d \log n} = n^{-d\ell}.
\]
This finishes the proof of the claim.
    \end{claimproof}
    \fi

    Take an index $i \in [\log n]$ satisfying \cref{claim:smallN}.
    We know that with probability at least
    $1 - n^{-d\ell}$ there exists $j \in [d \log n]$ such that the
    set $N_{i, j}$ contains more than $12\ell > \ell$ but less than $72\ell = \ell'$ vertices.
    The corresponding query to the data structure $\promisenl_{i, j}$ will return at least $\ell$ elements of $N_{i, j} \subseteq N$ with probability at least $1 - \Oh(n^{-d})$.
    Hence, by the union bound, the probability of incorrectly answering the query $\oper{sampleEdges}(a)$ is upper bounded by $n^{-d\ell} + \Oh(n^{-d}) \subseteq \Oh(n^{-d})$.
    The running time is dominated by the total running time of queries on all structures $\promisenl_{i, j}$ which is bounded by
    $\log n \cdot d \log n \cdot \Ohtilde(k\ell' \cdot d) \subseteq \Ohtilde(k\ell \cdot d^2)$.
\end{proof}

%\km{All consecutive calls to $\oper{sampleEdges}(a)$ will return the same subset of neighbors of $a$, that is, it is not exactly \emph{uniform sampling}. Should we clarify it somewhere?}
%\anka{our sampling is not random sampling at all but i would not clarify this.}

\section{Lifting the promise}
\label{sec:edges-inside}

In the previous section we presented data structures that work under the promise that $\splittance_G(A, B) \leq k$.
Now, we focus on the general case when no such condition is guaranteed to hold.
The main responsibility of our next data structure, $\nsample[n, k, d]$, is to implement the same functionality as data structures $\promisenl$ and $\promisens$, but to work without the assumption that $\splittance_G(A, B) \leq k$ and, additionally, to offer methods for listing non-edges with both endpoints in $A$ and edges with both endpoints in $B$ (these were needed for queries in the data structures $\promisenl$ and $\promisens$).
%;
% these sets are required by the data structures from \cref{sec:edges-across}.
All those methods are required to give correct outputs only when the inequality $\splittance_G(A,B) \leq k$ holds, but the data structure should persist also when $\splittance_G(A,B)>k$, and return to giving correct answers again once the splittance decreases to at most $k$.

%If the current graph satisfies this condition, but its splittance has been greater than $k$ over last $r$ updates, we will show how to compute the new sets in time proportional to $r$ which will imply acceptable \emph{amortized} update time.

%\km{Konkurs na najlepszą nazwę poniższej struktury danych :)}\anka{oj tam oj tam, jest ok.}

\begin{lemma}\label{lem:listing2}
    There exists a data structure $\nsample[n, k, d]$ parametrized by integers
    $k, d \in \N$ that runs on a
    dynamic graph $G = (V, E)$ with $V=[n]$ and
    maintains a partition $(A, B)$ of $V$ satisfying
    $\splittance_G(A, B) = \splittance(G)$ by supporting the
    following operations:
    \begin{itemize}
        \item $\oper{initialize}(n, k, d)$: sets $V \coloneqq [n]$,
        fixes the values of
        $k$ and $d$ for the entire run, initalizes
        $G \coloneqq (V, \emptyset)$ and
        $(A, B) \coloneqq (\emptyset, V)$.
        Runs in time $\Ohtilde(kd^2 \cdot n)$.
        \item $\oper{update}(uv)$:
        inserts edge $uv$ if $uv \not\in E(G)$
        or removes edge $uv$ if $uv \in E(G)$. Runs
        in $\Ohtilde(kd(k + d))$ amortized time.
        \item $\oper{splittance}()$:
        returns $\splittance(G)$ in time $\Oh(1)$.
    \end{itemize}
    Moreover, if  $\splittance(G) \leq k$,
    the data structure answers the following queries:
    \begin{itemize}
        \item $\oper{listNonEdgesA}()$: returns the list of all non-edges of $G$ with both endpoints in $A$;
        \item $\oper{listEdgesB}()$: returns the list of all edges of $G$ with both endpoints in $B$;
        \item $\oper{sampleEdges}(a)$: given a vertex $a \in A$, returns a subset
        $
            \NBsample(a) \subseteq N_G(a) \cap B
        $
        of size $\min(10k, |N_G(a) \cap B|)$;
        \item $\oper{sampleNonEdges}(b)$: given a vertex $b \in B$, returns a subset
        $
            \NAsample(b) \subseteq \Non_G(b) \cap A
        $
        of size $\min(10k, |\Non_G(b) \cap A|)$.
    \end{itemize}
    The queries $\oper{listNonEdgesA}$ and $\oper{listEdgesB}$ run in time $\Ohtilde(k)$, and the queries $\oper{sampleEdges}$ and $\oper{sampleNonEdges}$ run in time $\Ohtilde(k^2d^2)$.
    The data structure is randomized and returns correct answers
    with probability $1 - \Oh(n^{-d})$.
\end{lemma}

\begin{proof}
    The data structure $\nsample[n, k, d]$ stores:
    \begin{itemize}
        \item the adjacency list representation of $G$ and
        representation of $A$ and $B$ as sets;
        \item the values of $k$ and $d$;
        \item an instance $\dsplit[n]$ of the data structure from \cref{lem:splittance} running on $G$;
        \item an instance $\promisenl[n, k, \ell, d', S]$ of the data structure from \cref{lem:listing1}, where $\ell = k$, $d' = d+3$ and $S = V(G)$;
        \item an instance $\promisens[n, k, \ell', d]$ of the data structure from \cref{lem:sampling}, where $\ell' = 10k$;
        \item a set $\oper{verticesUpd} \subseteq V(G)$ and
        sets
        $\oper{edgesUpd}, \oper{edgesB},
        \oper{nonEdgesA} \subseteq \binom{V(G)}{2}$
        described in the next paragraph.
    \end{itemize}
    
    Let $G_0, G_1, \ldots, G_t$ be a sequence of consecutive graphs
    maintained by our data structure, that is,
    for every $i = 0, 1, \ldots, t-1$, $G_{i+1}$ is obtained from $G_i$
    by adding or removing a single edge.
    Moreover, let $(A_i, B_i)$ be the corresponding partition
    of  $V(G_i)$ maintained by the data structure,
    and let $E_i$ be the set of edges of $G_i$, for $i = 0, \ldots, t-1$.
    Define $s \coloneqq \max\ \{ 0 \leq i \leq t \mid \splittance(G_i) \leq k \}$.
    Note that $s$ is well-defined as the initial graph $G_0$ is a
    split graph.
    Then, sets $\oper{verticesUpd},
        \oper{edgesUpd}, \oper{edgesB},
        \oper{nonEdgesA}$
    are defined by the following invariants that must be satisfied at all times:
    \begin{itemize}
        \item \iEdge $\oper{edgesB}$ is the set of all edges of the graph $G_s[B_s]$, and $\oper{nonEdgesA}$ is the set of all non-edges of the graph $G_s[A_s]$, both with probability $1 - \Oh(n^{-d})$;
        \item \iUpd $\oper{verticesUpd} = A_s \symd A_t = B_s \symd B_t$, and $\oper{edgesUpd} = E_s \symd E_t$;
        \item \iProm the data structures $\promisenl$ and $\promisens$ store the graph $G_s$ with the partition $(A_s, B_s)$ of its vertices. 
    \end{itemize}
    In other words, sets $\oper{edgesB}$ and $\oper{nonEdgesA}$ store the answers to the queries $\oper{listEdgesB}()$ and $\oper{listNonEdgesA}()$ that were correct at the last time step when the splittance was bounded by $k$. Sets $\oper{verticesUpd}$ and $\oper{edgesUpd}$ store all the updates that got accumulated since this time step.
    Also, note that the data structures $\promisenl$ and
    $\promisens$ run on the graph $G_s = (A_s \uplus B_s, E_s)$,
    not $G_t$, as they require that the stored graph $G$ and its
    partition $(A,B)$ satisfy
    $\splittance_G(A,B) \leq k$.

    Let us observe at this point that one can assume that the inequality $s \leq n^2$ holds.
    This comes from the fact that after  $n^2$ updates, we can recompute the correct values of $\oper{edgesB}$ and $\oper{nonEdgesA}$ for $G_s$ (and renumber $G_0 \coloneqq  G_s$) which amortizes to $\Ohtilde(1)$ cost per update.
    
    \subparagraph*{Initialization.}
    \iflipics
    Since the initialization of \nsample\xspace is straightforward, we moved the detailed description to \cref{sec:app-edges-inside}.
    \else
    \begin{itemize}
    \item We initialize the graph $G=([n],\emptyset)$
     and $(A,B)=(\emptyset,V)$. We set $\nonEdgesA = \edgesB = \emptyset$.
    %\item We save the values of $k$ and $d$.
    \item The data structures $\dsplit[n]$, $\promisenl[n, k, \ell, d', V(G)]$ and $\promisens[n, k, \ell', d]$ are initialized with the corresponding parameters.
    \item We set $\oper{edgesB}$, $\oper{nonEdgesA}$,
    $\oper{verticesUpd}$ and $\oper{edgesUpd}$ to be empty sets.
\end{itemize}
The running time is dominated by the initialization of all the stored data structures which takes time
\[
    \Ohtilde(n) + \Ohtilde(\ell d' \cdot n) + \Ohtilde(\ell' d^2 \cdot n) \subseteq \Ohtilde(kd^2 \cdot n).
\]
    \fi
    
    \subparagraph*{Update.}
    Consider an operation $\oper{update}(uv)$, where $uv$ is an edge to
    be inserted or deleted.
    First, we pass the same update to the inner structure $\dsplit$ which returns a subset of vertices $\Vmoved$ of size $\Oh(1)$ which need to be moved between $A$ and $B$.
    Then, we can update the sets $A$ and $B$, accordingly.
    Next, we update sets $\oper{verticesUpd}$ and
    $\oper{edgesUpd}$ as follows:\\
    $
        \oper{verticesUpd}' = \oper{verticesUpd} \symd \Vmoved \quad \text{ and } \quad \oper{edgesUpd}' = \oper{edgesUpd} \symd \{ e \}.
    $

    Let $G_s, G_{s+1}, \ldots, G_t = G$ be the suffix of the sequence of consecutive graphs given to the data structure, where $s$ is the previous moment our dynamic graph $G$ was of splittance at most $k$.
    If $\splittance(G_t) > k$, we can finish the update as all the invariants \iEdge, \iUpd and \iProm hold.
    The running time is $\Ohtilde(1)$.

    Now, assume that $\splittance(G_t) \leq k$.
    Then $\max\ \{ 0 \leq i \leq t \mid \splittance(G_i) \leq k \} = t$, and we need to perform additional modifications so that the invariants \iEdge, \iUpd and \iProm hold.
    Let $r \coloneqq t - s$ be the number of updates that were applied since the last time our graph was of splittance at most $k$.
    Our goal now is to recompute the sets $\edgesB$ and $\nonEdgesA$, so that the invariant \iEdge holds, in time proportional to $r$.
    First, let us make a~useful observation.

    \iflipics
    \begin{claim}[$\app$]\label{claim:intersection}
    \else
    \begin{claim}\label{claim:intersection}
    \fi
        \claimintersectiontext
    \end{claim}
    \newtheorem*{claimintersection}{\cref{claim:intersection}}
    \iflipics
    \else
    \begin{claimproof}
    Denote $m \coloneqq  |A_s \cap B_t|$.
Since $\splittance(G_s) \leq k$, there were at most $k$ non-edges in the subgraph $G_s[A_s]$, and consequently there were at least $\binom{m}{2} - k$ edges in $G_s[A_s \cap B_t]$.
On the other hand, we also have $\splittance(G_t) \leq k$, and thus there are at most $k$ edges in $G_t[B_t]$, and in particular, at most $k$ edges in $G_t[A_s \cap B_t]$.

Starting from $G_s$ we removed at most $r$ edges until we reached $G_t$, hence
$
    r \geq \left( \binom{m}{2} - k \right) - k.
$
This implies that $m \leq \Oh(\sqrt{r + k})$ as desired.
The second inequality is analogous.
    \end{claimproof}
    \fi

    Now, let us focus on updating the set $\edgesB$, so that after the update it reflects time step $t$ instead of $s$ (the set $\nonEdgesA$ is handled analogously).
    In other words, we need to locate all the edges of $G_t$ with both endpoints in $B_t$.
    Denote this new set by $\edgesB'$.
    All such edges can be classified into the following three
    categories (see \cref{fig:edgecases} for illustration). In the following, we use variables $\edgesB$
    and $\nonEdgesA$ as they are in the data structure before
    updating, i.e., their content refers to the time step $s$.
    %\mp{Figure of partitions $(A_s,B_s)$ and $(A_t,B_t)$ with different types of edges?}

\begin{figure}[h]
\centering
\begin{tikzpicture}[scale=1.2]
   \tikzstyle{vertex}=[circle,fill=black,minimum size=0.15cm,inner sep=0pt]
   \tikzstyle{edge}=[very thick]
   \tikzstyle{impedge}=[very thick,red]
   \tikzstyle{impnonedge}=[very thick,dashed,black]

        \fill[orange!20] (-2,-1.3) rectangle (0,2);

        \fill[blue!20] ( 2,-1.3) rectangle (0,2);
        \node at ( 1,2.3) {$B_t$};

        \node at (-1,2.3) {$A_t$};
        \fill[orange!40] ( -2,0.3) rectangle (0,2);

        \fill[blue!40] ( 2,0.3) rectangle (0,2);
        \node at (-2.4,1.1) {$A_s$};
        \node at (-2.4,-0.6) {$B_s$};

        \node[vertex] (x) at (0.3, 1.5) {};
        \node[above] at (x) {$a_1$};
        \node[vertex] (y) at (0.8, 1.5) {};
        \node[above] at (y) {$a_2$};

        \node[vertex] (xx) at (0.3, 1.2) {};
        \node[below] at (0, 1.6) {$(2)$};
        \node[vertex] (yy) at (0.6, 1.2) {};

        \node[vertex] (z) at (0.3,-0.8) {};
        \node[below] at (z) {$b_1$};
        \node[vertex] (t) at (0.8,-0.8) {};
        \node[below] at (t) {$b_2$};
        \node[vertex] (zz) at (0.3,-0.5) {};
        \node[above] at (zz) {$(1)$};
        \node[vertex] (tt) at (0.6,-0.5) {};

        \node[vertex] (q) at (1.7, 1.5) {};
        \node[above] at (q) {$a_1$};
        \node[vertex] (p) at (1.7,-0.8) {};
        \node[below] at (p) {$b_1$};
        \node[vertex] (qq) at (1.4, 1.2) {};
        \node[vertex] (pp) at (1.4,-0.5) {};
        \node at (2,0.3) {$(3)$};

        \node[vertex] (qqq) at (0.7, 0.6) {};
        \node[vertex] (ppp) at (0.7,0) {};
        \node[vertex] (qqqq) at (0.9, 0.6) {};
        \node[vertex] (pppp) at (0.9,0) {};
        \node at (0.65,0.75) {$\Emod_1$};

        \node[vertex] (w) at (0.1, 0.4) {};
        \node[vertex] (e) at (-0.5,-0.3) {};
        \node[vertex] (r) at (0.3, 0.4) {};
        \node[vertex] (u) at (-0.3,-0.3) {};
        \node at (-0.35,0.35) {$\Emod_2$};

        \draw[edge] (x) -- (y);
        \draw[edge] (xx) -- (yy);
        \draw[edge] (z) -- (t);
        \draw[edge] (zz) -- (tt);
        \draw[edge] (p) -- (q);
        \draw[edge] (pp) -- (qq);
        \draw[impnonedge] (ppp) -- (qqq);
        \draw[impnonedge] (pppp) -- (qqqq);
        \draw[impnonedge] (w) -- (e);
        \draw[impnonedge] (r) -- (u);

\end{tikzpicture}
\caption{Edge types needed for recomputing $\oper{edgesB}$ in the proof of Lemma~\ref{lem:listing2}.}
\label{fig:edgecases}
\end{figure}
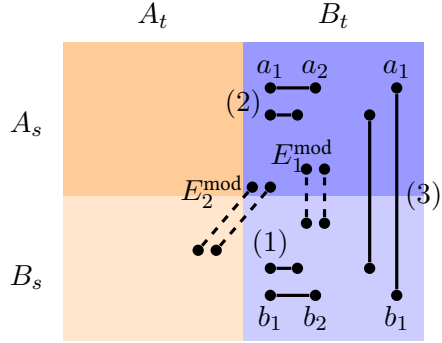

    \begin{itemize}
        \item[(1)] Edges $b_1 b_2$, where $b_1, b_2 \in B_s \cap B_t$.
        
        Observe that every edge in $G_t$ of this form is either
        \begin{itemize}
            \item an edge of $G_s[B_s]$, but there were at most $k$ such ones as $\splittance(G_s) \leq k$; or
            \item an edge inserted during an update after time step $s$, but there are at most $r$ such ones.
        \end{itemize}
        So using the sets $\edgesB$ and $\oper{edgesUpd}$ we can enumerate all such edges in time~$\Ohtilde(k + r)$.

        \item[(2)] Edges $a_1 a_2$, where $a_1, a_2 \in A_s \cap B_t$.
        
        We can list the set $A_s \cap B_t \subseteq \oper{verticesUpd}$ in time $\Ohtilde(r)$, and by \cref{claim:intersection} we can check all the pairs of vertices of $A_s \cap B_t$ in time $\Ohtilde(r + k)$.

        \item[(3)] Edges $a_1 b_1$, where $a_1 \in A_s \cap B_t$ and $b_1 \in B_s \cap B_t$.
        
        To list all such edges we start with applying on $\promisenl$ a temporary modification by calling
        \[
            \oper{update}(\emptyset,\ \Emod,\ \nonEdgesA,\ \edgesB),
        \] where $\Emod=\Emod_1 \cup \Emod_2$ consists of:
        \begin{itemize}
            \item $\Emod_1$: all edges of $\oper{edgesUpd}$ of the form $ab$, where $a \in A_s \cap B_t$ and $b \in B_s \cap B_t$; there are at most $r$ such edges.
            \item $\Emod_2$: all edges of the graph $G_s$ of the form $ab$, where $a \in A_s \cap B_t$ and $b \in B_s \cap A_t$; we can list these edges brutally in time $\Ohtilde(r+k)$ by \cref{claim:intersection}.
        \end{itemize}
        That is, structure $\promisenl$ works now on the graph $\Gtmp = (A_s \uplus B_s, \Etmp)$ such~that
        \begin{itemize}
            \item the graphs $\Gtmp[A_s]$ and $\Gtmp[B_s]$ are isomorphic $G_s[A_s]$ and $G_s[B_s]$, respectively,
            \item there are no edges between $A_s \cap B_t$ and $A_t \cap B_s$, and
            \item the edges between $A_s \cap B_t$ and $B_s \cap B_t$ coincide in the graphs $\Gtmp$ and $G_t$.
        \end{itemize}
        The cost of this temporary update is $\Ohtilde(k + \ell d'(k + r)) = \Ohtilde(kd(k + r))$.
        Observe that $\edgesB$ and $\nonEdgesA$ store currently
        respectively the edges within $B_s$ and non-edges within
        $A_s$ in $\Gtmp$, and thus $\splittance_{\Gtmp}(A_s,B_s) \leq k$,
        as required by the $\promisenl$ data structure.

        Now, we iterate over all vertices $a \in A_s \cap B_t$;
        recall that their number is bounded by $r$ as they are all
        contained in $\oper{verticesUpd}$.
        Recall that $\splittance_{G_t}(A_t, B_t) \leq k$, hence
        $a$ has at most $k$ neighbors in $B_s \cap B_t$, and we
        would like to list them all.
        However, this can be done precisely by calling
        $\oper{listNeighborsBS}(a)$ on $\promisenl$.
        The total cost of such queries for
        all $a_1 \in A_s \cap B_t$ is
        \[
            \Ohtilde(|A_s \cap B_t| \cdot k \ell d') \leq \Ohtilde(r k^2 d).
        \]
        Finally, we revert $\promisenl$ to the state before the temporary update, by calling
        \[
            \oper{update}(\emptyset,\ \Emod,\ \nonEdgesA,\ \edgesB).
        \]
    \end{itemize}
    Summing up, the amortized cost of computing the set $\edgesB'$ (and $\nonEdgesA'$) is upper-bounded by
    \[
        \frac{\Ohtilde(k + r) + \Ohtilde(r + k) + \Ohtilde(k d (k + r)) + \Ohtilde(r k^2d)}{r} \subseteq \Ohtilde(k^2 d).
    \]
    Now, we need to bound the probability of incorrectly computing the sets $\edgesB'$ and $\nonEdgesA'$.
    In the process of recomputation we do at most $\Oh(n)$ queries to $\promisenl$, each having an error probability $\Oh(n^{-d'}) = \Oh(n^{-(d+3)})$.
    Hence, with probabilty $1 - \Oh(n^{d+2})$ all these queries return correct answers.
    Furthermore, we rely on the fact that the sets $\edgesB$ and $\nonEdgesA$ were computed correctly for the graph $G_s$.
    Recall that we can assume $s \leq n^2$, and thus by the union bound these sets were computed correctly for all the graphs $G_0, \ldots, G_t$ with probability at least
    \[
        1 - n^2 \cdot \Oh(n^{-(d+2)}) = 1 - \Oh(n^{-d}).
    \]
    
    Next, we should update the structures $\promisenl$ and $\promisens$ so that their inner states correspond to the graph $G_t$.
    To this end, we run on both structures the update
    \[
        \oper{update}(\oper{verticesUpd},\ \oper{edgesUpd},\ \nonEdgesA',\ \edgesB').
    \]
    Then, the invariant \iProm holds, and the amortized cost of these updates is
    \[
        \frac{\Ohtilde(k + r \ell d') + \Ohtilde(k d + r\ell' d^2 )}{r} \subseteq \Ohtilde(k d^2).
    \]

    Finally, we set $\oper{verticesUpd}' \coloneqq \emptyset$ and $\oper{edgesUpd}' \coloneqq \emptyset$, so that invariant \iUpd holds.
    
    \subparagraph*{Queries.}
    \iflipics
    Each query can be answered by doing a single call to one of the structures $\dsplit$ or $\promisens$.
    The details are given in \cref{sec:app-edges-inside}.
    \else
    To give the current value of $\splittance(G)$, it is enough to call $\oper{splittance}()$ on $\dsplit$.
Now, assume that $\splittance(G) \leq k$.
Consequently, by the invariant \iEdge, all the inner variables store up-to-date values for the current graph $G = G_t$.
Hence, to answer queries $\oper{listEdgesB}()$ and $\oper{listNonEdgesA}()$, it is enough to return the current sets $\oper{edgesB}$ and $\oper{nonEdgesA}$, respectively.
Finally, to answer queries $\oper{sampleEdges}(a)$ for $a \in A$ and $\oper{sampleNonEdges}(b)$ for $b \in B$, it is enough to do the same calls on the structure $\promisens$ whose inner state is up-to-date according to the invariant \iProm.

    \fi 
    \end{proof}

\section{Localizing the obstructions}
\label{sec:obstacles}

Let us recall that the graph is split if and only if it is
$\Fc$-free, where $\Fc = \{2K_2, C_4, C_5\}$.
We call the members of ${\Fc}$ simply {\em{obstructions}}.
In this section, using the interface of the data structure $\nsample[n, k, d]$ from \cref{lem:listing2}, we show how to extend this data structure with a method that is able to either conclude that $G$ is split, or report an obstruction in $G$ witnessing that it is not split.
The following statement encapsulates this result.

\begin{lemma}\label{lem:obstacles}
    The data structure $\nsample[n, k, d]$
    of \cref{lem:listing2} can be extended by a~method
    \begin{itemize}
        \item $\oper{findObstruction}()$ which returns \ansSplit if $\splittance(G) = 0$; or a subset of vertices $U \subseteq V(G)$ such that the induced subgraph $G[U]$ is isomorphic to one of the graphs of $\Fc = \{ 2K_2, C_4, C_5 \}$ if $0 < \splittance(G) \leq k$.
    \end{itemize}
    The query runs in time $\Ohtilde(k^{\Oh(1)} \cdot d^2)$ and is correct with probability $1-\Oh(n^{c-d})$ for some constant $c$ whenever $\splittance(G) \leq k$; without this assumption, there are no guarantees on the correctness.
\end{lemma}

\begin{proof}
We assume that $\splittance(G) = \splittance_G(A, B) \leq k$, otherwise we can provide any answer.
%\km{Przeczytajcie akapit poniżej, nie wiem, czy jest wystarczająco zrozumiały.}\anka{tak, to jest bardzo dziwne}
%Note that in the description below we will only refer to the interface of $\nsample[n, k, d]$, not to its inner implementation.
%However, for technical reasons, we assume here that the data structure $\nsample$ was internally initialized with a parameter $d' \coloneqq d + \delta$ instead of $d$, for some constant $\delta \leq 10$ to be set later.
%This does not affect asymptotic bounds on the running time but increases slightly the probability of returning correct answers by $\nsample$,
%and we require it since in the implementation of $\oper{findObstruction}(\cdot)$ we can make up to $\Oh(k^{\Oh(1)})$ calls to the interface of $\nsample$.
In what follows, we use extensively the following two
observations. %(Claim~\ref{claim:sub1} and Claim~\ref{claim:sub2}).

\begin{claim}\label{claim:sub1}
    \claimsubtext
\end{claim}
\begin{claimproof}
Using the calls $\oper{sampleEdges}(\cdot)$ on $\nsample$ in
time $\Ohtilde(k^2 d^2)$ we can extract sets $N_1$ and
$N_2$, where $N_t$ contains $\min(|N_G(a_t)\cap B|,10k)$
neighbors of $a_t$ in $B$, for $t\in \{1,2\}$.
If $|N_1| \leq 3\sqrt{k}$ or $|N_2| \leq 3\sqrt{k}$,
then the second case holds and we can report this together with
the corresponding set of at most $3\sqrt{k}$ neighbors in $B$.
So suppose otherwise: both $|N_1| > 3\sqrt{k}$ and $|N_2| > 3\sqrt{k}$.

We now consider two cases.
First, suppose that $|N_1 \cap N_2| > 2\sqrt{k}$.
Recall that $\splittance(G) \leq k$, and thus there are at
most $k$ edges in graph $G[B]$.
Since the graph $G[N_1 \cap N_2]$ has more than
$
    \frac{2\sqrt{k} \cdot (2\sqrt{k} - 1)}{2} = 2k - \sqrt{k} \geq k
$
pairs of vertices, and $N_1 \cap N_2 \subseteq B$, there must be a non-edge $b_1b_2 \not\in E(G)$ such that $b_1, b_2 \in N_1 \cap N_2$.
Then, one can observe that the induced subgraph $G[\{a_1, a_2, b_1, b_2\}]$ is isomorphic to $C_4$. It is easy to see that this reasoning can be turned into an $\Ohtilde(k)$-time algorithm to find $b_1,b_2$.

Otherwise, if $|N_1 \cap N_2| \leq 2\sqrt{k}$, then $|N_1 \setminus N_2| > \sqrt{k}$ and $|N_2 \setminus N_1| > \sqrt{k}$.
Again, using the fact there are at most $k$ edges in the graph $G[B]$, there must be a non-edge $b_1b_2 \not\in E(G)$ such that $b_1 \in N_1 \setminus N_2$ and $b_2 \in N_2 \setminus N_1$.
Then, we observe that the induced subgraph $G[\{a_1, a_2, b_1, b_2\}]$ is isomorphic to $2K_2$. Again, such $b_1,b_2$ can be found in time $\Ohtilde(k)$.
\end{claimproof}

The second observation is symmetric to the first one.
We include the statement for clarity.

\begin{claim}
\label{claim:sub2}
    Let $b_1, b_2 \in B$ be vertices such that $b_1b_2 \in E(G)$. Then, at least one of the following holds:
    \begin{itemize}
        \item there is a subset $U \subseteq V(G)$ such that $\{b_1, b_2\} \subseteq U$ and $G[U]$ is isomorphic to $2K_2$ or $C_4$;
        \item at least one of the vertices $b_1, b_2$ has at most $3\sqrt{k}$ non-neighbors in $A$.
    \end{itemize}
    Moreover, in time $\Ohtilde(k^2 d^2)$ we can return a
    case that holds, together with a witnessing set $U$ for the
    first case, or the non-neighborhood of $b_1$ or $b_2$ in
    $A$ for the second case. The answer is correct with probability
    at least $1 - \Oh(n^{-d})$.
\end{claim}

We are ready to proceed with the description of the
$\oper{findObstruction}(\cdot)$ query.
Suppose now that there is a set $U \subseteq V(G)$ such that
$G[U]$ is an obstruction.
We examine all three cases for $G[U] \in \{2K_2, C_4, C_5\}$, and
in each case we provide a procedure that, assuming the existence of
$U$ as above, detects some set $U'$ (possibly different from $U$)
that induces an obstruction. If no case yields an obstruction,
we conclude that $G$ is split.
It will be always clear that the procedure for locating $U'$ runs
in time $\Ohtilde(k^{\Oh(1)})$.
%, so we leave the easy verification to the reader.
%Let us fix some terminology used in the description of
%procedure $\oper{findObstruction}(\cdot)$.
%We often enumerate a set of candidates for some vertex $u$ and
%then {\em{fix}} $u$ among the candidates.
%By this we mean branching into as many subcases as candidates,
%in each fixing $u$ to be a different candidate,
%and proceeding with the reasoning under this assumption.\anka{to jest super niejasne jaki w ogole termin tu definiujemy}
%Next, a set shall be called {\em{small}} if its cardinality is at most $3\sqrt{k}$.
%Finally,
Recall that by running $\oper{listNonEdgesA}()$ and $\oper{listEdgesB}()$, we have access to the sets $\oper{nonEdgesA}$ containing all non-edges within $A$ and $\oper{edgesB}$ containing all edges within $B$, and each of these sets has cardinality at most $k$. For brevity, a set is {\em{small}} if its size is at most $3\sqrt{k}$.

\iflipics
\else

\begin{figure}[h]
\centering
\begin{tikzpicture}[scale=1.2]

   \tikzstyle{vertex}=[circle,fill=black,minimum size=0.15cm,inner sep=0pt]
   \tikzstyle{edge}=[very thick]
   \tikzstyle{impedge}=[very thick,red]
   \tikzstyle{impnonedge}=[very thick,dashed,red]

   \begin{scope}[shift={(-4.5,1)},scale=0.5]

        \fill[orange!20] (-2,-1.3) rectangle (0,2);
        \node at (-1,1.6) {$A$};
        \fill[blue!20] ( 2,-1.3) rectangle (0,2);
        \node at ( 1,1.6) {$B$};

        \node[vertex] (x) at (-1.5, 0.5) {};
        \node[above] at (x) {$x$};
        \node[vertex] (y) at (-0.5, 0.5) {};
        \node[above] at (y) {$y$};
        \node[vertex] (z) at (-1.5,-0.5) {};
        \node[below] at (z) {$z$};
        \node[vertex] (t) at (-0.5,-0.5) {};
        \node[below] at (t) {$t$};

        \draw[edge] (x) -- (y);
        \draw[edge] (z) -- (t);
        \draw[impnonedge] (x) -- (t);
        \draw[impnonedge] (z) -- (y);

        \node at (0,-2) {(a)};
   \end{scope}

   \begin{scope}[shift={(-2.25,1)},scale=0.5]

        \fill[orange!20] (-2,-1.3) rectangle (0,2);
        \node at (-1,1.6) {$A$};
        \fill[blue!20] ( 2,-1.3) rectangle (0,2);
        \node at ( 1,1.6) {$B$};

        \node[vertex] (x) at (-1.2, 0) {};
        \node[left] at (x) {$x$};
        \node[vertex] (y) at (-0.4, 0.8) {};
        \node[above] at (y) {$y$};
        \node[vertex] (z) at (-0.4,-0.8) {};
        \node[below] at (z) {$z$};
        \node[vertex] (t) at (0.4,0) {};
        \node[right] at (t) {$t$};

        \draw[edge] (x) -- (y);
        \draw[edge] (z) -- (t);
        \draw[impnonedge] (y) -- (z);
        \draw[impnonedge] (z) -- (x);

        \node at (0,-2) {(b)};
   \end{scope}

   \begin{scope}[shift={(0,1)},scale=0.5]
     \fill[orange!20] (-2,-1.3) rectangle (0,2);
        \node at (-1,1.6) {$A$};
        \fill[blue!20] ( 2,-1.3) rectangle (0,2);
        \node at ( 1,1.6) {$B$};

        \node[vertex] (x) at (-0.7, 0.7) {};
        \node[left] at (x) {$x$};
        \node[vertex] (y) at (-0.7, -0.7) {};
        \node[left] at (y) {$y$};
        \node[vertex] (z) at ( 0.7, 0.7) {};
        \node[right] at (z) {$z$};
        \node[vertex] (t) at ( 0.7, -0.7) {};
        \node[right] at (t) {$t$};

        \draw[edge] (x) -- (y);
        \draw[impedge] (z) -- (t);

        \node at (0,-2) {(c)};
   \end{scope}

   \begin{scope}[shift={(2.25,1)},scale=0.5]

        \fill[orange!20] (-2,-1.3) rectangle (0,2);
        \node at (-1,1.6) {$A$};
        \fill[blue!20] ( 2,-1.3) rectangle (0,2);
        \node at ( 1,1.6) {$B$};

        \node[vertex] (x) at (-0.7, 0.7) {};
        \node[left] at (x) {$x$};
        \node[vertex] (y) at ( 0.7,  0.7) {};
        \node[right] at (y) {$y$};
        \node[vertex] (z) at (-0.7, -0.7) {};
        \node[left] at (z) {$z$};
        \node[vertex] (t) at ( 0.7, -0.7) {};
        \node[right] at (t) {$t$};

        \draw[edge] (x) -- (y);
        \draw[edge] (z) -- (t);
        \draw[impnonedge] (z) -- (x);

        \node at (0,-2) {(d)};
   \end{scope}

   \begin{scope}[shift={(4.5,1)},scale=0.5]
     \fill[orange!20] (-2,-1.3) rectangle (0,2);
        \node at (-1,1.6) {$A$};
        \fill[blue!20] ( 2,-1.3) rectangle (0,2);
        \node at ( 1,1.6) {$B$};

        \node[vertex] (x) at ( 1.2, 0) {};
        \node[right] at (x) {$x$};
        \node[vertex] (y) at ( 0.4, 0.8) {};
        \node[above] at (y) {$y$};
        \node[vertex] (z) at ( 0.4,-0.8) {};
        \node[below] at (z) {$z$};
        \node[vertex] (t) at (-0.4,0) {};
        \node[left] at (t) {$t$};

        \draw[impedge] (x) -- (y);
        \draw[edge] (z) -- (t);

        \node at (0,-2) {(e)};
    \end{scope}

   \begin{scope}[shift={(-3.375,-1.5)},scale=0.5]
     \fill[orange!20] (-2,-1.3) rectangle (0,2);
        \node at (-1,1.6) {$A$};
        \fill[blue!20] ( 2,-1.3) rectangle (0,2);
        \node at ( 1,1.6) {$B$};

        \begin{scope}[shift={(-1,0)}]
        \node[vertex] (w) at (0:0.8) {};
        \node[right] at (w) {$w$};
        \node[vertex] (x) at (72:0.8) {};
        \node[above] at (x) {$x$};
        \node[vertex] (y) at (144:0.8) {};
        \node[left] at (y) {$y$};
        \node[vertex] (z) at (216:0.8) {};
        \node[left] at (z) {$z$};
        \node[vertex] (t) at (288:0.8) {};
        \node[below] at (t) {$t$};
        \draw[edge] (x) -- (y) -- (z) -- (t) -- (w) -- (x);
        \draw[impnonedge] (x) -- (z) -- (w) -- (y) -- (t) -- (x);
        \end{scope}

        \node at (0,-2) {(a)};
    \end{scope}

   \begin{scope}[shift={(-1.125,-1.5)},scale=0.5]
     \fill[orange!20] (-2,-1.3) rectangle (0,2);
        \node at (-1,1.6) {$A$};
        \fill[blue!20] ( 2,-1.3) rectangle (0,2);
        \node at ( 1,1.6) {$B$};

        \begin{scope}[shift={(-0.5,0)}]
        \node[vertex] (w) at (0:0.8) {};
        \node[right] at (w) {$w$};
        \node[vertex] (x) at (72:0.8) {};
        \node[above] at (x) {$x$};
        \node[vertex] (y) at (144:0.8) {};
        \node[left] at (y) {$y$};
        \node[vertex] (z) at (216:0.8) {};
        \node[left] at (z) {$z$};
        \node[vertex] (t) at (288:0.8) {};
        \node[below] at (t) {$t$};
        \draw[edge] (x) -- (y) -- (z) -- (t) -- (w) -- (x);
        \draw[impnonedge] (x) -- (z);
        \draw[impnonedge] (y) -- (t);
        \end{scope}

        \node at (0,-2) {(b)};
    \end{scope}
   \begin{scope}[shift={( 1.125,-1.5)},scale=0.5]
     \fill[orange!20] (-2,-1.3) rectangle (0,2);
        \node at (-1,1.6) {$A$};
        \fill[blue!20] ( 2,-1.3) rectangle (0,2);
        \node at ( 1,1.6) {$B$};

        \begin{scope}[shift={(-0.1,0)}]
        \node[vertex] (w) at (0+36:0.8) {};
        \node[right] at (w) {$w$};
        \node[vertex] (x) at (72+36:0.8) {};
        \node[above] at (x) {$x$};
        \node[vertex] (y) at (144+36:0.8) {};
        \node[left] at (y) {$y$};
        \node[vertex] (z) at (216+36:0.8) {};
        \node[below] at (z) {$z$};
        \node[vertex] (t) at (288+36:0.8) {};
        \node[right] at (t) {$t$};
        \draw[edge] (w) -- (x) -- (y) -- (z) -- (t);
        \draw[impedge] (t) -- (w);
        \draw[impnonedge] (x) -- (z);
        \end{scope}

        \node at (0,-2) {(c)};
    \end{scope}
   \begin{scope}[shift={(3.375,-1.5)},scale=0.5]
     \fill[orange!20] (-2,-1.3) rectangle (0,2);
        \node at (-1,1.6) {$A$};
        \fill[blue!20] ( 2,-1.3) rectangle (0,2);
        \node at ( 1,1.6) {$B$};

        \begin{scope}[shift={(-0.1,0)}]
        \node[vertex] (x) at (0+36:0.8) {};
        \node[right] at (x) {$x$};
        \node[vertex] (w) at (72+36:0.8) {};
        \node[above] at (w) {$w$};
        \node[vertex] (y) at (144+36:0.8) {};
        \node[left] at (y) {$y$};
        \node[vertex] (t) at (216+36:0.8) {};
        \node[below] at (t) {$t$};
        \node[vertex] (z) at (288+36:0.8) {};
        \node[right] at (z) {$z$};
        \draw[edge] (x) -- (y) -- (z) -- (t) -- (w) -- (x);
        \draw[impnonedge] (y) -- (w);
        \draw[impnonedge] (y) -- (t);
        \end{scope}

        \node at (0,-2) {(d)};
    \end{scope}

\end{tikzpicture}
\caption{Non-symmetric cases in $2K_2$ (top row) and $C_5$ (bottom row) localization. Edges and non-edges (dashed) marked in red are present on the lists $\edgesB$ and $\nonEdgesA$, respectively, hence they can be selected among $\Oh(k)$ candidates and fixed.}
\label{fig:cases}
\end{figure}
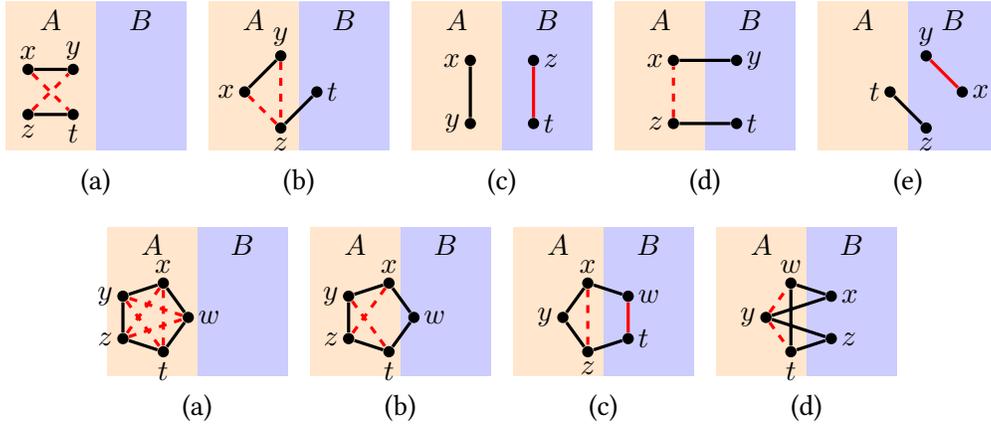

\fi

\subparagraph*{Localizing $2K_2$.} Suppose first that $G[U]=G[\{x,y,z,t\}]$ is isomorphic to $2K_2$, where $xy, zt \in E(G)$.
We consider all possible alignments of vertices $x, y, z, t$ with
respect to the partition $(A, B)$. We divide the analysis
into five cases depening on
the intersection $U \cap A$. See \cref{fig:cases} \iflipics in \cref{sec:app-obstacles} \else \fi for an illustration.

\begin{enumerate}
    \item[(a)] $|U \cap A| = 4$.
    
    Then, in particular, $xt$ and $yz$ are two non-edges within $A$.
    Hence, $U$ can be detected by exhaustive search through all the pairs of non-edges in the set $\nonEdgesA$.
    
    \item[(b)] $|U \cap A| = 3$.
    
    Without loss of generality assume that $U \cap A = \{x, y, z\}$.
    Since $xz$ and $yz$ are two non-edges with both endpoints in $A$ we can list all $\Oh(k^2)$ candidates for $(x, y, z)$ by inspecting all the pairs of non-edges in $\nonEdgesA$.
    Having $(x, y, z)$ fixed, it remains to locate (any) vertex $t' \in B$ that is a neighbor of $z$ and a non-neighbor of $x$ and $y$ (the original vertex $t\in U$ witnesses that such $t'$ exists).
    We call $\oper{sampleEdges}(z)$ on $\nsample$ to obtain a set $\NBsample(z)$ consisting of up to $10k$ neighbors of $z$ in $B$.
    If $|\NBsample(z)| < 10k$, then $\NBsample(z)$ contains all neighbors of $z$ in $B$, and we can test all of them for a vertex~$t'$.
    So assume now that $|\NBsample(z)| = 10k$.
    Then, we run the subroutine from \cref{claim:sub1} for pairs $(x, z)$ and $(y, z)$.
    If any of these calls detects and returns an obstruction, we are done.
    Otherwise, since $z$ has more than $3\sqrt{k}$ neighbors in $B$, both $x$ and $y$ must have at most $3\sqrt{k}$ neighbors in $B$ and the subroutine returns the corresponding sets.
    Denote them by $N_B(x)$ and $N_B(y)$.
    Since
    \[
        |\NBsample(z)| = 10k > 3\sqrt{k} + 3\sqrt{k} \geq |N_B(x)| + |N_B(y)|,
    \]
    there exists a vertex $t' \in \NBsample(z) \setminus (N_B(x) \cup N_B(y))$. Then $x,y,z,t'$ induce a $2K_2$.
    %Then, $G[\{x, y, z, t\}]$ forms a $2K_2$ graph, as desired.
    % if $z$ has at most $3 \sqrt{k}$ neighbors in $B$, we can list all of them with call $\oper{sampleEdges(z)}$ on $\Dsample$.
    
    \item[(c)] $|U \cap A| = 2$ and the vertices of $U \cap A$ are adjacent in $G$.
    
    Without loss of generality assume that $U \cap A = \{ x, y \}$ and $U \cap B = \{z, t\}$.
    Then $zt$ is an edge with both endpoints in $B$, hence using the set $\edgesB$ we can list $\Oh(k)$ candidates for $(z, t)$ and thus fix $(z,t)$.
    Next, we run the subroutine from \cref{claim:sub2} for $(z, t)$, and assume that it did not detect an obstruction (for otherwise we are done).
    Then, the subroutine discovers that one of $z$ or $t$ (say $z$) has small non-neighborhood $\Non_A(z)$ in $A$. Since $x, y \in \Non_A(z)$, it is enough to check at most $\binom{3\sqrt{k}}{2}=\Oh(k)$ candidates for $(x, y)$.

    \iflipics
    \else
    \item[(d)] $|U \cap A| = 2$ and the vertices of $U \cap A$ are not adjacent in $G$.
    
Without loss of generality assume that $U \cap A = \{x, z\}$.
We see that $xz$ is a non-edge with both endpoints in $A$, thus we can list $\Oh(k)$ candidates for $(x, z)$ and  fix $(x,z)$. Next, we run the subroutine from \cref{claim:sub1} for $(x, z)$, and assume that it did not detect an obstruction, for otherwise we are done.
Without loss of generality, the subroutine returns a small neighborhood $N_B(x)$ of $x$ in $B$.
We inspect the set $N_B(x)$ for all valid candidates for $y$ (non-neighbors of $z$) and thus fix $y$.
It remains to locate any vertex $t'$ such that $zt' \in E(G)$ and $xt', yt' \not\in E(G)$ (the original vertex $t\in U$ witnesses that such $t'$ exists).
To this end, we call $\oper{sampleEdges}(z)$ on $\nsample$.
Denote the returned set by $\NBsample(z)$.
If $|\NBsample(z)| < 10k$, it contains all neighbors of $z$ in $B$, and we can test them all for a vertex $t'$.
Otherwise, we have that
\[
    |\NBsample(z)| = 10k > 3\sqrt{k} + k \geq |N_B(x)| + |\edgesB| \geq |N_B(x)| + |N_B(y)|,
\]
and hence there must be a vertex $t' \in \NBsample(z)$ which is a non-neighbor of both $x$ and~$y$. Then $x,y,z,t'$ induce a $2K_2$.

\item[(e)] $|U \cap A| = 1$.

Without loss of generality assume that $U \cap A = \{ t \}$.
Then $U \cap B = \{x, y, z\}$, and since $xy \in E(G)$, we can list $\Oh(k)$ candidates for $(x, y)$ using the set $\edgesB$, and thus fix $(x,y)$.
Next, we call the subroutine from \cref{claim:sub2} for $(x, y)$.
Assuming no obstruction is found, it returns a small non-neighborhood of $x$ or $y$ in $A$, and we can inspect it to list $\Oh(\sqrt{k})$ candidates for a vertex $t$.
Having $(x, y, t)$ fixed, it remains to find any vertex $z'\in B$ that is a neighbor of $t$ and a non-neighbor of $x$ and $y$ (the original $z\in U$ witnesses that such a $z$ exists).
We call $\oper{sampleEdges}(t)$ and denote the obtained set by $\NBsample(t)$.
Similarly as before, if $|\NBsample(t)| < 10k$, then $\NBsample(t)$ it contains all neighbors of $t$ in $B$, and we can test all of them for a vertex $z;$.
Otherwise, we see that
\[
    |\NBsample(t)| = 10k > 2|\edgesB| \geq |N_B(x)| + |N_B(y)|,
\]
hence there must be a vertex $z' \in \NBsample(t)$ which is a non-neighbor of both $x$ and $y$. Then $x,y,z',t$ induce a $2K_2$.

\item[(f)] $|U \cap A| = 0$.

Then, $xy$ and $zt$ are two edges within $B$, hence we can detect them by an exhaustive search through all the pairs of edges in the set $\edgesB$.
    \fi
\end{enumerate}

\iflipics
Due to the space constraints, we present the remaining cases in \cref{sec:app-obstacles}.
\else
\fi

\subparagraph*{Localizing $C_4$.}
Again, suppose that $G[U]=G[\{x,y,z,t\}]$ is isomorphic to~$C_4$, where $xy, zt \not\in E(G)$.
Note that $C_4$ is the complement of $2K_2$, and all our data structures are symmetric under taking the complement and swapping the sides $(A, B)$.
So the analysis of possible arrangements of $C_4$ is fully symmetric to the analysis of $2K_2$, and hence we omit it.

\subparagraph*{Localizing $C_5$.}
\iflipics
Since the analysis of possible arrangements of $C_5$ is done similarly to the case of $2K_2$, we present it in \cref{sec:app-obstacles}.
\bigskip
\else
Suppose $G[U]=G[\{x, y, z, t, w\}]$ is isomorphic to $C_5$, where $xy, yz, zt, tw, wx \in E(G)$.
We consider possible alignments of $\{x, y, z, t, w\}$ with respect to the partition $(A, B)$.
\begin{itemize}
    \item[(a)] $|U \cap A| = 5$.
    
    Then every vertex of $U$ is an endpoint of some non-edge within $A$, so we can detect $U$ by examining all 5-tuples of non-edges in $\nonEdgesA$.

    \item[(b)] $|U \cap A| = 4$.
    
    Without loss of generality we assume that $U \cap A = \{x, y, z, t\}$.
    We see that $xz, yt \not\in E(G)$, hence by iterating over the set $\nonEdgesA$ we can find $\Oh(k^2)$ candidates for $(x, y, z, t)$, and thus fix $(x,y,z,t)$.
    Next, we run the subroutine from \cref{claim:sub1} for $(x, t)$.
    Assuming no obstruction is detected, it returns a small neighborhood of $x$ or $t$ in $B$, and thus we obtain $\Oh(\sqrt{k})$ candidates for a vertex $w \in B$ which must be a neighbor of both $x$ and $t$.
    Investigating those candidate reveals an induced $C_5$.

    \item[(c)] $|U \cap A| = 3$ and the vertices of $U \cap B$ are adjacent in $G$.
    
    Without loss of generality we assume that $U \cap A = \{x, y, z\}$.
    Observe that $xz \not\in E(G)$ and $tw \in E(G)$, and thus by iterating over the sets $\nonEdgesA$ and $\edgesB$ we can list $\Oh(k^2)$ candidates for $(x, z, t, w)$, and thus fix $(x, z, t, w)$. Then we run the subroutine from \cref{claim:sub2} for $(t, w)$.
    Assuming no obstruction is detected, it returns a small non-neighborhood of $t$ or $w$ in $A$, and thus we obtain $\Oh(\sqrt{k})$ candidates for a vertex $y \in A$ which must be a non-neighbor of both $t$ and $w$. Investigating those candidate reveals an induced $C_5$.

    \item[(d)] $|U \cap A| = 3$ and the vertices of $U \cap B$ are not adjacent in $G$.
    
    Without loss of generality we assume that $U \cap B = \{x, z\}$ and $U \cap A = \{y, t, w\}$.
    We see that $yt, yw \not\in E(G)$, hence by searching through the set $\nonEdgesA$ we can list $\Oh(k^2)$ candidates for $(y, t, w)$, and thus fix $(y,t,w)$.
    Now note that $x$ is a common neighbor of $y$ and $w$ in $B$, and $z$ is a common neighbor of $y$ and $t$ in $B$.
    Therefore, to list small sets of candidates for $x$ and $z$, it is enough to run the subroutine from \cref{claim:sub1} for $(y, w)$ and $(y, t)$, respectively.
    Each such call either directly returns an obstruction, or discovers that one of the vertices has small neighborhood in $B$, and this neighborhood can be used as a set of $\Oh(\sqrt{k})$ candidates for $x$ or $z$. Examining all such pairs of candidates reveals an induced $C_5$.
    
    \item[(e)] $|U \cap A| = 2$ and the vertices of $U \cap A$ are adjacent in $G$.
    
    Observe that after complementing the graph $G[U]$ and swapping the sides $(A, B)$, this case becomes the previously considered case (d).
    Since all our inner data structures are closed under this operation, the analysis of this case is fully analogous and we omit it.
    
    % Without loss of generality, say that $U \cap A = \{t, w\}$ and $U \cap B = \{x, y, z\}$.
    % Then $xy$ and $yz$ are edges within $B$, hence using the set $\edgesB$ we can list all $\Oh(k^2)$ candidates for $(x, y, z)$.
    
    \item[(f)] $|U \cap A| = 2$ and the vertices of $U \cap A$ are non-adjacent in $G$.
    
    Similarly as before, this case is analogous to the already considered case (c).
    
    % Without loss of generality, say that $U \cap A = \{x, z\}$ and $U \cap B = \{y, t, w\}$.
    % Then $xz$ is a non-edge within $A$, and $tw$ is an edge within $B$, therefore we can list all $\Oh(k^2)$ candidates for $(x, z, t, w)$.
    % Observe that vertex $y$ is a common neighbor of $x$ and $z$ in $B$, hence by running the subroutine from \cref{claim:sub1} for $(x, z)$ we either obtain an obstruction, or a neighborhood of $x$/$z$ in $B$ of size $\Oh(\sqrt{k})$ which can be used as a set of candidates for $y$.
    
    \item[(g)] $|U \cap A| = 1$.
    
    This case is analogous to the case (b).
    
    % Without loss of generality, say that $U \cap A = \{w\}$.
    % Observe that $xy$ and $zt$ are edges of $G$ within $B$, hence we can inspect the set $\edgesB$ for all $\Oh(k^2)$ candidates for $(x, y, z, t)$.
    % It remains to locate a vertex $w$.
    % To this end, we run the subroutine from \cref{claim:sub2} for $(y, z)$.
    % If no obstruction is found, it returns a small non-neighborhood of $y$ or $z$ in $A$, and we can test all the vertices in it for a valid vertex $w$, as there are at most $\Oh(\sqrt{k})$ candidates for $w$.
    
    \item[(h)] $|U \cap A| = 0$.
    
    This case is analogous to the case (a).
    % Then, every vertex of $U$ is an endpoint of some edge within $B$, therefore using the set $\edgesB$ we can examine all such cases in time $\Ohtilde(k^5)$.
\end{itemize}
\fi

This finishes the case study.
Let us summarize the implementation of the query
$\oper{findObstruction}()$.
Every such query makes single calls $\oper{listNonEdgesA}()$
and $\oper{listEdgesB}()$, $k^{\Oh(1)}$ calls $\oper{sampleEdges}(\cdot)$ and $\oper{sampleNonEdges}(\cdot)$ (either directly, or via the calls to the subroutines from \cref{claim:sub1} and \cref{claim:sub2}), and spends $\Ohtilde(k^{\Oh(1)})$ time on internal computation.
Therefore, the total running time of $\oper{findObstruction}()$ is bounded by
\[
    k^{\Oh(1)} \cdot \Ohtilde(k^2 d^2) + \Ohtilde(k^{\Oh(1)}) \subseteq \Ohtilde(k^{\Oh(1)} \cdot d^2).
\]
Further, as we may assume that $k\leq n^2$, the probability of
incorrectly answering the query is:
\[
    \Oh\left(n^{-d} \right) + \Oh(k^{\Oh(1)}) \cdot \Oh \left( n^{-d}\right) = \Oh\left(n^{\Oh(1)} \cdot n^{-d}\right) \subseteq \Oh\left(n^{c-d}\right).
\]
for some fixed constant $c$.
\end{proof}

\iflipics
With \cref{lem:obstacles} established, the proof of \cref{thm:main} follows easily: Assuming we maintain the data structure $\nsample$ extended as in \cref{lem:obstacles}, upon each query one simply executes the standard branching algorithm for {\sc{Split Completion}}, using method $\oper{findObstruction}()$ for finding consecutive obstructions and $\oper{update}(\cdot)$ for modifying the graph on the fly during backtracking. Details can be found in \cref{sec:branching}.
\else
\fi

\iflipics
\else
\section{Branching}
\label{sec:branching}

We are ready to prove our main result, \cref{thm:main}, which we recall for convenience.

\begin{mainthm}
\mainthmtext
\end{mainthm}

\begin{proof}
Our main data structure stores:
\begin{itemize}
    \item the values of $k$ and $d$;
    \item an instance $\nsample[n, k, d']$ of the data structure
    from \cref{lem:listing2}, extended with a method
    from \cref{lem:obstacles}, for $d' \coloneqq d+k+c$, where $c$ is
    the constant from \cref{lem:obstacles}.
\end{itemize}

\subparagraph*{Initialization.}
We set the values of $k$ and $d$ for the entire run, and initialize the inner data structure $\nsample$ with corresponding values.
The total running time of this phase is
\[
    \Ohtilde(k \cdot d'^2 \cdot n) = \Ohtilde(k^{\Oh(1)} \cdot d^2 \cdot n).
\]

\subparagraph*{Update.}
Consider an edge update $e = uv$ (that is, we either insert $e$ to the graph $G$, or remove it).
First, we pass this to the structure $\nsample$.
Let $G'$ denote the newly obtained graph.
Now, we should update the answer to our query: we need to determine whether one can insert at most $k$ edges to $G'$ so that the resulting graph is split.

If $\splittance(G') > k$, then we know that $(G',k)$ is a \noinstance of the \textsc{Split Completion} problem.
Assume then that $\splittance(G') \leq k$.
Now we make a query $\oper{findObstruction}()$ on $\nsample$.
If $\oper{findObstruction}()$ reports the answer \ansSplit, then $G'$ is a split graph, and in particular, $(G', k)$ is a~\yesinstance of the \textsc{Split Completion} problem.
Otherwise, $\oper{findObstruction}()$ returns a subset $U \subseteq V(G')$ such that $G'[U]$ is an obstruction.

Since in the \textsc{Split Completion} problem we are only allowed to insert new edges, we know that in order to make the graph $G'$ split, we need to insert to it at least one non-edge of $G'[U]$. Note that $G'[U]$ has at most $5$ non-edges (in the case of $C_5$).
We branch into all these possibilities, and when considering an insertion of a non-edge $e \in \overline{E}(G'[U])$, we temporarily insert it to the graph $G'$ (by doing appropriate call to the structure $\nsample$), and proceed recursively, where at each step either we obtain that the current graph is split, or $\nsample$ localizes an obstruction which needs to be eliminated from the graph by subsequent insertions.
Since we are allowed to insert at most $k$ edges to $G'$, we can exit branches of depth larger than $k$, and thus the total number of nodes in the search tree (and also the number of calls to $\nsample$) is bounded by $\Oh(5^k)$.

The correctness of this recursive procedure is straightforward under an assumption that all the calls to $\nsample$ return correct answers.
The amortized running time of the entire update is bounded by
\[
    \Oh(5^k) \cdot \Ohtilde(k^{\Oh(1)} \cdot d'^2) = \Ohtilde(5^k \cdot k^{\Oh(1)} \cdot d^2)
\]
By the union bound, the probability of reporting an incorrect answer is at most
\[
    \Oh\left(\frac{5^k}{n^{d'-c}}\right) = \Oh\left(\frac{5^k}{n^{k + d+c-c}}\right) \subseteq \Oh\left(\frac{1}{n^d}\right).
\]
Finally, let us remark that this is in fact a Monte Carlo algorithm: once our procedure finds a set of edges $F \subseteq \binom{V(G)}{2}$ to be inserted, we can verify deterministically whether $G'' = (V(G'), E(G') \cup F)$ is a split graph by temporarily inserting edges of $F$ to $\nsample$, checking whether $\splittance(G'') = 0$, and removing edges of $F$ from $\nsample$ afterwards.
\end{proof}

We remark that the proof of \cref{thm:main} can be trivially adjusted to handle the {\sc{Split Edge Deletion}} problem instead: once an obstruction is located we branch into at most five ways of breaking it by an edge deletion instead of an edge insertion. The running time guarantees remain the same.

\fi

\section{Conclusions}\label{sec:conclusions}

In this work we gave a dynamic data structure for the {\sc{Split Completion}} problem with amortized update time $2^{\Oh(k)}\cdot (d\log n)^{\Oh(1)}$, where the error probability is bounded by $1-\Oh(n^{-d})$. Two natural questions that one may ask are the following:
\begin{itemize}
 \item Can the data structure be derandomized? The usage of color coding in \cref{lem:listing1} can be derandomized in a standard way using universal hash families. However, we also apply randomization in a non-trivial way in the proof of \cref{lem:sampling}, and this application we do not know how to derandomize.
 \item A careful analysis of our arguments reveals that the polylogarithmic factor in the obtained update time is $(\log n)^4$.
 Can this be improved? Could one expect update time independent of $n$, that is, of the form $f(k)$ for some computable function $f$?
\end{itemize}
As for amortization, the only place where amortization comes into play is the proof of \cref{lem:listing2}: the updates are queued during the time when the splittance is above $k$, and then they get executed all at once only when the splittance becomes bounded by $k$. Consequently, our data structure actually executes updates in {\em{worst-case time}} $2^{\Oh(k)}\cdot (d\log n)^{\Oh(1)}$ under the assumption that the splittance is never larger than $k$. We believe that relying on amortization to handle the general setting without any promises is a rather small hindrance. Consequently, while the question about removing amortization is definitely relevant, we consider it less interesting than the two questions stated above.

In a wider context, our result together with the previous work of Iwata and Oka~\cite{IwataO14} and of Alman et al.~\cite{AlmanMW20} suggest that investigating parameterized dynamic data structures for graph modification problems might turn out to be a particularly fruitful direction. Concrete problems that we suspect to be amenable to dynamization are {\sc{Planarization}}~\cite{MarxS12,JansenLS14} and various graph modification problems connected to chordal graphs and their subclasses such as interval, proper interval, trivially perfect, or threshold graphs; these were studied extensively, see e.g.~\cite{BliznetsFPP15,BliznetsFPP18,Cao16,Cao17,CaoM16,DrangeFPV15,DrangeP18,FominSV13,FominV13,KeCOLW18,Marx10,HofV13,VillangerHPT09}.

\bibliography{references}

\begin{thebibliography}{10}

\bibitem{AlmanMW20}
Josh Alman, Matthias Mnich, and Virginia~Vassilevska Williams.
\newblock Dynamic parameterized problems and algorithms.
\newblock {\em {ACM} Trans. Algorithms}, 16(4):45:1--45:46, 2020.

\bibitem{AlonYZ95}
Noga Alon, Raphael Yuster, and Uri Zwick.
\newblock Color-coding.
\newblock {\em J. {ACM}}, 42(4):844--856, 1995.

\bibitem{BliznetsFPP15}
Ivan Bliznets, Fedor~V. Fomin, Marcin Pilipczuk, and Micha\l{} Pilipczuk.
\newblock A subexponential parameterized algorithm for {P}roper {I}nterval
  {C}ompletion.
\newblock {\em {SIAM} J. Discret. Math.}, 29(4):1961--1987, 2015.

\bibitem{BliznetsFPP18}
Ivan Bliznets, Fedor~V. Fomin, Marcin Pilipczuk, and Micha\l{} Pilipczuk.
\newblock Subexponential parameterized algorithm for {I}nterval {C}ompletion.
\newblock {\em {ACM} Trans. Algorithms}, 14(3):35:1--35:62, 2018.

\bibitem{Cao16}
Yixin Cao.
\newblock Linear recognition of almost interval graphs.
\newblock In {\em 27th Annual {ACM-SIAM} Symposium on Discrete Algorithms,
  {SODA} 2016}, pages 1096--1115. {SIAM}, 2016.

\bibitem{Cao17}
Yixin Cao.
\newblock Unit interval editing is fixed-parameter tractable.
\newblock {\em Inf. Comput.}, 253:109--126, 2017.

\bibitem{CaoM16}
Yixin Cao and D{\'{a}}niel Marx.
\newblock Chordal editing is fixed-parameter tractable.
\newblock {\em Algorithmica}, 75(1):118--137, 2016.

\bibitem{ChenCDFHNPPSWZ21}
Jiehua Chen, Wojciech Czerwi\'nski, Yann Disser, Andreas~Emil Feldmann, Danny
  Hermelin, Wojciech Nadara, Marcin Pilipczuk, Micha\l{} Pilipczuk, Manuel
  Sorge, Bart\l{}omiej Wr{\'{o}}blewski, and Anna Zych{-}Pawlewicz.
\newblock Efficient fully dynamic elimination forests with applications to
  detecting long paths and cycles.
\newblock In {\em 2021 {ACM-SIAM} Symposium on Discrete Algorithms, {SODA}
  2021}, pages 796--809. {SIAM}, 2021.

\bibitem{platypus}
Marek Cygan, Fedor~V. Fomin, \L{}ukasz Kowalik, Daniel Lokshtanov, D{\'{a}}niel
  Marx, Marcin Pilipczuk, Micha\l{} Pilipczuk, and Saket Saurabh.
\newblock {\em Parameterized Algorithms}.
\newblock Springer, 2015.

\bibitem{CyganP13}
Marek Cygan and Marcin Pilipczuk.
\newblock Split {V}ertex {D}eletion meets {V}ertex {C}over: New fixed-parameter
  and exact exponential-time algorithms.
\newblock {\em Inf. Process. Lett.}, 113(5-6):179--182, 2013.

\bibitem{DrangeFPV15}
P{\aa}l~Gr{\o}n{\aa}s Drange, Fedor~V. Fomin, Micha\l{} Pilipczuk, and Yngve
  Villanger.
\newblock Exploring the subexponential complexity of completion problems.
\newblock {\em {ACM} Trans. Comput. Theory}, 7(4):14:1--14:38, 2015.

\bibitem{DrangeP18}
P{\aa}l~Gr{\o}n{\aa}s Drange and Michal Pilipczuk.
\newblock A polynomial kernel for {T}rivially {P}erfect {E}diting.
\newblock {\em Algorithmica}, 80(12):3481--3524, 2018.

\bibitem{DvorakKT14}
Zdenek Dvo\v{r}{\'{a}}k, Martin Kupec, and Vojtech T\r{u}ma.
\newblock A dynamic data structure for {MSO} properties in graphs with bounded
  tree-depth.
\newblock In {\em 22th Annual European Symposium on Algorithms, {ESA} 2014},
  volume 8737 of {\em Lecture Notes in Computer Science}, pages 334--345.
  Springer, 2014.

\bibitem{DvorakT13}
Zdenek Dvo\v{r}{\'{a}}k and Vojtech T\r{u}ma.
\newblock A dynamic data structure for counting subgraphs in sparse graphs.
\newblock In {\em 13th International Symposium on Algorithms and Data
  Structures, {WADS} 2013}, volume 8037 of {\em Lecture Notes in Computer
  Science}, pages 304--315. Springer, 2013.

\bibitem{FominSV13}
Fedor~V. Fomin, Saket Saurabh, and Yngve Villanger.
\newblock A polynomial kernel for {P}roper {I}nterval {V}ertex {D}eletion.
\newblock {\em {SIAM} J. Discret. Math.}, 27(4):1964--1976, 2013.

\bibitem{FominV13}
Fedor~V. Fomin and Yngve Villanger.
\newblock Subexponential parameterized algorithm for {M}inimum {F}ill-in.
\newblock {\em {SIAM} J. Comput.}, 42(6):2197--2216, 2013.

\bibitem{FoldesH77}
Stéphane Földes and Peter~L. Hammer.
\newblock Split graphs.
\newblock In {\em Eighth Southeastern Conference on Combinatorics, Graph Theory
  and Computing}, volume XIX of {\em Congressus Numerantium}, pages 311--315,
  1977.

\bibitem{GhoshK0MPRR15}
Esha Ghosh, Sudeshna Kolay, Mrinal Kumar, Pranabendu Misra, Fahad Panolan,
  Ashutosh Rai, and M.~S. Ramanujan.
\newblock Faster parameterized algorithms for deletion to split graphs.
\newblock {\em Algorithmica}, 71(4):989--1006, 2015.

\bibitem{GrezMPPR22}
Alejandro Grez, Filip Mazowiecki, Micha\l{} Pilipczuk, Gabriele Puppis, and
  Cristian Riveros.
\newblock Dynamic data structures for timed automata acceptance.
\newblock {\em Algorithmica}, 84(11):3223--3245, 2022.

\bibitem{TheSplittance}
Peter~L. Hammer and Bruno Simeone.
\newblock The splittance of a graph.
\newblock {\em Combinatorica}, 1:275--284, 1981.

\bibitem{HeggernesM09}
Pinar Heggernes and Federico Mancini.
\newblock Dynamically maintaining split graphs.
\newblock {\em Discret. Appl. Math.}, 157(9):2057--2069, 2009.

\bibitem{Ibarra08}
Louis Ibarra.
\newblock Fully dynamic algorithms for chordal graphs and split graphs.
\newblock {\em {ACM} Trans. Algorithms}, 4(4):40:1--40:20, 2008.

\bibitem{IwataO14}
Yoichi Iwata and Keigo Oka.
\newblock Fast dynamic graph algorithms for parameterized problems.
\newblock In {\em 14th Scandinavian Symposium and Workshops on Algorithm
  Theory, {SWAT} 2014}, volume 8503 of {\em Lecture Notes in Computer Science},
  pages 241--252. Springer, 2014.

\bibitem{JansenLS14}
Bart M.~P. Jansen, Daniel Lokshtanov, and Saket Saurabh.
\newblock A near-optimal planarization algorithm.
\newblock In {\em 25th Annual {ACM-SIAM} Symposium on Discrete Algorithms,
  {SODA} 2014}, pages 1802--1811. {SIAM}, 2014.

\bibitem{KeCOLW18}
Yuping Ke, Yixin Cao, Xiating Ouyang, Wenjun Li, and Jianxin Wang.
\newblock Unit interval vertex deletion: Fewer vertices are relevant.
\newblock {\em J. Comput. Syst. Sci.}, 95:109--121, 2018.

\bibitem{KorhonenMNPS23}
Tuukka Korhonen, Konrad Majewski, Wojciech Nadara, Micha\l{} Pilipczuk, and
  Marek Soko\l{}owski.
\newblock Dynamic treewidth.
\newblock {\em CoRR}, abs/2304.01744, 2023.
\newblock To appear in the proceedings of FOCS 2023.

\bibitem{LewisY80}
John~M. Lewis and Mihalis Yannakakis.
\newblock The node-deletion problem for hereditary properties is {NP}-complete.
\newblock {\em J. Comput. Syst. Sci.}, 20(2):219--230, 1980.

\bibitem{MajewskiPS23}
Konrad Majewski, Micha\l{} Pilipczuk, and Marek Soko\l{}owski.
\newblock Maintaining {CMSO}$_2$ properties on dynamic structures with bounded
  feedback vertex number.
\newblock In {\em 40th International Symposium on Theoretical Aspects of
  Computer Science, {STACS} 2023}, volume 254 of {\em LIPIcs}, pages
  46:1--46:13. Schloss Dagstuhl --- Leibniz-Zentrum f{\"{u}}r Informatik, 2023.

\bibitem{Marx10}
D{\'{a}}niel Marx.
\newblock Chordal deletion is fixed-parameter tractable.
\newblock {\em Algorithmica}, 57(4):747--768, 2010.

\bibitem{MarxS12}
D{\'{a}}niel Marx and Ildik{\'{o}} Schlotter.
\newblock Obtaining a planar graph by vertex deletion.
\newblock {\em Algorithmica}, 62(3-4):807--822, 2012.

\bibitem{NatanzonSS01}
Assaf Natanzon, Ron Shamir, and Roded Sharan.
\newblock Complexity classification of some edge modification problems.
\newblock {\em Discret. Appl. Math.}, 113(1):109--128, 2001.

\bibitem{OlkowskiPRWZ23}
J\k{e}drz\k{e}j Olkowski, Micha\l{} Pilipczuk, Mateusz Rychlicki, Karol
  W\k{e}grzycki, and Anna Zych{-}Pawlewicz.
\newblock Dynamic data structures for parameterized string problems.
\newblock In {\em 40th International Symposium on Theoretical Aspects of
  Computer Science, {STACS} 2023}, volume 254 of {\em LIPIcs}, pages
  50:1--50:22. Schloss Dagstuhl --- Leibniz-Zentrum f{\"{u}}r Informatik, 2023.

\bibitem{HofV13}
Pim {van 't Hof} and Yngve Villanger.
\newblock Proper interval vertex deletion.
\newblock {\em Algorithmica}, 65(4):845--867, 2013.

\bibitem{VillangerHPT09}
Yngve Villanger, Pinar Heggernes, Christophe Paul, and Jan~Arne Telle.
\newblock Interval {C}ompletion is fixed parameter tractable.
\newblock {\em {SIAM} J. Comput.}, 38(5):2007--2020, 2009.

\end{thebibliography}

\end{document}